\documentclass[12pt]{article}
\usepackage{amsmath, amssymb, amsfonts}
\usepackage{amsthm}
\usepackage[utf8]{inputenc}
\usepackage{csquotes}
\usepackage{geometry}
\usepackage{adjustbox}
\usepackage{picture}
\usepackage{graphicx}
\usepackage{verbatim}
\usepackage{array}
\usepackage{enumitem}
\usepackage{booktabs}
\usepackage{hyperref}
\usepackage{url}

\title{Fractal Analysis on the Real Interval: A Constructive Approach via Fractal Countability}
\author{Stanislav Semenov \\
\href{mailto:stas.semenov@gmail.com}{stas.semenov@gmail.com} \\
\href{https://orcid.org/0000-0002-5891-8119}{ORCID: 0000-0002-5891-8119}}
\date{March 30, 2025}

\theoremstyle{definition}
\newtheorem{definition}{Definition}[section]

\newtheorem{example}{Example}[section]
\newtheorem{conjecture}[definition]{Conjecture}

\theoremstyle{plain}
\newtheorem{theorem}[definition]{Theorem}

\newtheorem{corollary}[definition]{Corollary}
\newtheorem{proposition}[definition]{Proposition}

\theoremstyle{remark}
\newtheorem*{remark}{Remark}

\begin{document}

\maketitle

\begin{abstract}
    This paper develops a technical and practical reinterpretation of the real interval \([a,b]\) under the paradigm of fractal countability. Instead of assuming the continuum as a completed uncountable totality, we model \([a,b]\) as a layered structure of constructively definable points, indexed by a hierarchy of formal systems. We reformulate classical notions from real analysis---continuity, measure, differentiation, and integration---in terms of stratified definability levels \(S_n\), thereby grounding the analytic apparatus in syntactic accessibility rather than ontological postulation. The result is a framework for \emph{fractal analysis}, in which mathematical operations are relativized to layers of expressibility, enabling new insights into approximation, computability, and formal verification.
\end{abstract}

\subsection*{Mathematics Subject Classification}
03F60 (Constructive and recursive analysis), 26E40 (Constructive analysis), 03F03 (Proof theory and constructive mathematics)

\subsection*{ACM Classification}
F.4.1 Mathematical Logic, F.1.1 Models of Computation

\section*{Introduction}

The real interval \([a,b]\) is one of the most fundamental objects in mathematics, serving as the stage for calculus, topology, measure theory, and functional analysis. In classical mathematics, this interval is treated as a subset of the real line \(\mathbb{R}\), whose points form an uncountable continuum constructed via the power set \(\mathcal{P}(\mathbb{N})\). However, as critically examined in \cite{Semenov2025FractalBoundaries}, this assumption carries substantial ontological costs---including impredicative comprehension, nonconstructive existence, and the problematic totality of infinite sets.

Constructive mathematics offers an alternative, focusing on objects that can be explicitly defined, computed, or represented within formal systems. Yet within this setting, the interval \([a,b]\) can no longer be treated as an unstructured totality. Only those points generated by syntactic processes are meaningfully accessible---and since formal systems are built from countable syntax, only a countable fragment of \([a,b]\) becomes available at each stage, as first quantified in \cite{Semenov2025FractalCountability}.

This paper systematically analyzes \([a,b]\) through the lens of \emph{fractal countability}, extending the layered model of constructive expressibility developed in \cite{Semenov2025FractalCountability}. Here, each formal system \(\mathcal{F}_n\) defines a countable subset \(S_n \subset [a,b]\) of points definable within that system. The full constructive trace of the interval emerges as \(\bigcup_{n=0}^\infty S_n\)---a countable but unbounded approximation to the continuum that preserves the meta-formal hierarchy introduced in our earlier work.

We demonstrate how classical analytical notions transform under this stratification:
\begin{itemize}
    \item Continuity becomes relative to \(S_n\)-topology, with layer transitions formalizing definability constraints
    \item Compactness requires syntactic coverings, resolving the tension identified in \cite{Semenov2025FractalBoundaries} between classical and constructive covering principles
    \item Differentiability and integrability explicitly track layer transitions \(S_n \to S_{n+k}\)
    \item Lebesgue measure decomposes into \(S_n\)-constructible approximations
\end{itemize}

New concepts like \emph{fractal continuity} and \emph{definitional neighborhoods} operationalize the expressibility hierarchy from \cite{Semenov2025FractalCountability}, while our treatment of pathological sets (e.g., Liouville numbers) addresses the definability boundary cases analyzed in \cite{Semenov2025FractalBoundaries}.

This work culminates in \emph{fractal analysis}---a technical framework where mathematical objects are not merely approximated, but \emph{positioned} within a definitional hierarchy. By grounding analysis in syntactic accessibility rather than ontological postulation, we provide formal tools for the stratified constructivity proposed in our earlier meta-theoretical investigations.

\section{Related Work and Foundational Context}
The mathematical landscape offers multiple approaches to reconciling classical analysis with constraints of definability and computation. While these frameworks share constructive underpinnings, they differ fundamentally in their treatment of the continuum.

\paragraph{Intuitionistic Analysis (Brouwer)}
Brouwer's intuitionism~\cite{Brouwer1927} conceives the continuum as a \emph{medium of free becoming}, where real numbers are generated by \emph{choice sequences} rather than pre-existing as completed sets. This view emphasizes temporal construction over static ontology, but resists formal codification and lacks explicit stratification of definability levels.

\paragraph{Constructive Analysis (Bishop)}
Bishop's framework~\cite{Bishop1967} provides a formalizable alternative, building real numbers as Cauchy sequences with moduli of convergence. However, it assumes a uniform continuum of constructively admissible points, without internal mechanisms to track \emph{varying degrees} of definability across formal systems.

\paragraph{Computable Analysis (Turing--Weihrauch)}
In computable analysis~\cite{PourElRichards1989,Weihrauch2000}, reals are defined via computable sequences, and functions are studied through effective continuity. While this yields a rigorous computational semantics, it fixes representations \emph{a priori} and does not account for the \emph{dynamic expansion} of definable objects under syntactic extensions.

\paragraph{Our Approach: Fractal Countability} 
Unlike these models, our theory interprets the continuum as a \emph{layered definitional horizon}, where each stratum \(S_n\) corresponds to a formal system \(\mathcal{F}_n\). This dynamic stratification aligns in spirit with reverse mathematics~\cite{Simpson2009} but prioritizes \emph{syntactic accessibility} over proof-theoretic strength. By embedding the continuum within a meta-formal hierarchy, we capture both its constructive content and its open-ended definability structure.

\section{Fractal Hierarchy on the Real Interval}

\subsection{Motivating Example: The Base Layer \texorpdfstring{\( S_1 \)}{S1}}
\label{subsec:base-layer}

Consider the definable fragment \( S_1 \cap [0,1] \) as a paradigm for our approach. Let:
\begin{itemize}
    \item \( \mathcal{F}_0 \) be a system of primitive recursive arithmetic,
    \item \( \mathcal{F}_1 \) extend \( \mathcal{F}_0 \) with computable root-finding for algebraic equations.
\end{itemize}

Then the corresponding layers satisfy:
\begin{itemize}
    \item \( S_0 \cap [0,1] \): Dyadic rationals \( \frac{p}{2^k} \) with primitive recursive \( p,k \).
    \item \( S_1 \cap [0,1] \): Algebraic numbers \( \alpha \in [0,1] \) with computable minimal polynomials.
    \item \emph{Transcendentals exclusion}: Constants like \( \pi \) or \( e \) require stronger systems (e.g., \( \mathcal{F}_3 \) for Taylor series definitions).
\end{itemize}

This illustrates how formal expressibility bounds the "reachable" subset of \([0,1]\). Crucially, operations like \( x \mapsto x^2 \) preserve \( S_1 \)-membership, enabling closed-form analysis within each layer.

\subsection{Formal Definition: The Inductive Hierarchy \texorpdfstring{\( S_n \)}{Sn}}
\label{subsec:inductive-hierarchy}

We now define the full hierarchy via a sequence of formal systems \( \{\mathcal{F}_n\}_{n \in \mathbb{N}} \):

\begin{definition}
For each \( n \geq 0 \), the \emph{\(n\)-th definability layer} \( S_n \subset [a,b] \) is:
\[
S_n := \{ x \in [a,b] \mid x \text{ is } \mathcal{F}_n\text{-definable} \},
\]
where \( \mathcal{F}_{n+1} \) conservatively extends \( \mathcal{F}_n \) with:
\begin{itemize}
    \item Additional comprehension (e.g., arithmetic \( \Delta^1_1 \)), 
    \item Higher-type recursion,
    \item Controlled choice principles.
\end{itemize}
\end{definition}

\paragraph{Key Properties}
\begin{itemize}
    \item \textbf{Closure}: If \( f \in \mathcal{F}_n \) is total and \( x \in S_n \), then \( f(x) \in S_n \).
    \item \textbf{Strict growth}: \( S_{n+1} \setminus S_n \neq \emptyset \) (e.g., \( \mathcal{F}_{n+1} \) defines Chaitin’s constants~\cite{Chaitin1975}).
    \item \textbf{Countable limit}: \( S^\infty := \bigcup_{n=0}^\infty S_n \) is a countable dense subset of \([a,b]\).
    \item \textbf{Relativity}: \( S_n \) depends on the syntactic capabilities of \( \mathcal{F}_n \), making the hierarchy meta-formal rather than purely mathematical.
\end{itemize}

\paragraph{Example Systems}
\begin{itemize}
    \item \( \mathcal{F}_0 = \mathrm{RCA}_0 \) (computable reals)~\cite{Simpson2009},
    \item \( \mathcal{F}_1 = \mathrm{ACA}_0 \) (arithmetic comprehension)~\cite{Simpson2009},
    \item \( \mathcal{F}_\omega \) (transfinite recursion up to \( \omega \)),
    \item \( \mathcal{F}_{n+1} = \mathcal{F}_n + \Sigma^1_n\text{-}\mathrm{IND} \).
\end{itemize}

This hierarchy induces a \emph{definability topology} on \([a,b]\), where neighborhoods are \( S_n \)-definable open sets. In Sections~\ref{sec:continuity-hierarchy}--\ref{sec:measure}, we develop analysis relative to this structure.

\section{Fractally Definable Points on \texorpdfstring{$[a,b]$}{[a,b]}}
\label{sec:fractal-points}

Classical analysis treats the interval $[a,b]$ as an uncountable completed whole, containing non-constructive elements. We reinterpret it through \emph{fractal countability} as a stratified union of definable fragments:
\[
S^\infty \cap [a,b] = \bigcup_{n=0}^\infty (S_n \cap [a,b]),
\]
where each $S_n$ contains points definable in formal system $\mathcal{F}_n$. This hierarchy captures both syntactic accessibility and structural complexity across definability layers.

\subsection{Local Definability and Density}

Each \( S_n \cap [a,b] \) is a countable dense subset of \([a,b]\), but its internal structure reflects the expressive limitations of \( \mathcal{F}_n \). For a given point \( x \in [a,b] \), we define:

\begin{definition}[Left/Right Definable Neighborhoods]
Given a level \( n \), define the one-sided neighborhoods:
\[
\mathcal{L}_\delta^n(x) = \{ y \in S_n \mid x - \delta < y < x \}, \quad
\mathcal{R}_\delta^n(x) = \{ y \in S_n \mid x < y < x + \delta \}.
\]
We say that \( x \) has a \emph{definable neighborhood} in \( S_n \) if both \( \mathcal{L}_\delta^n(x) \) and \( \mathcal{R}_\delta^n(x) \) are nonempty for all rational \( \delta > 0 \).
\end{definition}

This yields three \emph{local definability types} for \( x \in [a,b] \) relative to \( S_n \):
\begin{itemize}
    \item \emph{Isolated}: There exists \( \delta > 0 \) such that both \( \mathcal{L}_\delta^n(x) = \emptyset \) and \( \mathcal{R}_\delta^n(x) = \emptyset \).
    \item \emph{Dense}: For all \( \delta > 0 \), both neighborhoods are nonempty.
    \item \emph{One-sided gap}: Only one of \( \mathcal{L}_\delta^n(x) \) or \( \mathcal{R}_\delta^n(x) \) is empty for some \( \delta > 0 \).
\end{itemize}

\subsection{Classification by \texorpdfstring{\(S_n\)-Membership}{Sn-Membership}}
\label{subsec:classification}

We can classify classical real numbers by the lowest \( n \) such that they belong to \( S_n \):

\begin{itemize}
    \item \( \mathbb{Q} \cap [a,b] \subset S_0 \)
    \item Algebraic numbers (with computable minimal polynomials) lie in \( S_1 \)
    \item Computable reals (limits of recursive sequences of rationals) lie in \( S_2 \)
    \item Liouville numbers may require higher \( n \), depending on the definability of their approximations
\end{itemize}

\begin{definition}
The \emph{definability rank} of a real number \( x \in [a,b] \), relative to base system \( \mathcal{F}_0 \), is the least \( n \) such that \( x \in S_n \).
\end{definition}

This creates a natural stratification of \([a,b]\), with increasingly "complex" reals appearing at higher definability levels.

\begin{example}
The Liouville constant \( \ell = \sum_{k=1}^\infty 10^{-k!} \)~\cite{Liouville1844} exemplifies a real number with extremely rapid rational approximations. It is non-algebraic and not computable in low \( \mathcal{F}_n \), but may become definable in \( \mathcal{F}_3 \) if the system includes formal factorial growth and series convergence.
\end{example}

\subsection{Definability Gaps and Measure}

Despite the density of each \( S_n \), the complement \( [a,b] \setminus S_n \) is uncountable and exhibits measure-theoretic structure:

\begin{definition}[Definability Gap]
Let \( x \in [a,b] \). A \emph{definability gap at level \( n \)} is an interval \( (x-\delta, x+\delta) \) such that \( (x-\delta, x+\delta) \cap S_n = \emptyset \).
\end{definition}

Such gaps cannot occur in \( S^\infty \), but may persist for finite \( n \), and their distribution reflects the expressive reach of \( \mathcal{F}_n \). The set of all such gaps (i.e., the complement of \( S_n \)) is non-measurable in the classical sense but can be analyzed via \( S_n \)-measure:

\[
\mu_n([a,b] \setminus S_n) = \sup \left\{ \sum \ell(I_k) : \{I_k\} \text{ cover } [a,b] \setminus S_n,\ I_k \text{ with } S_n\text{-endpoints} \right\}.
\]

This shows that each finite level \( S_n \) captures no measurable content in classical terms: \( \mu_n([a,b] \setminus S_n) = b - a \), despite its density. However, the asymptotic union \( S^\infty \) satisfies \( \mu([a,b] \setminus S^\infty) = 0 \), reflecting convergence toward full measure. This underlines the idea that the hierarchy asymptotically exhausts \([a,b]\), though never fully within any single layer.

\section{Fractal Topology on the Interval}
\label{sec:topology}

Having established the hierarchy \( \{S_n\} \) of definable subsets of \([a,b]\), we now introduce a topology that reflects their stratified structure. Rather than relying on arbitrary open intervals, we define a family of neighborhoods that are themselves \( S_n \)-definable and capture local expressibility.

\subsection{\texorpdfstring{\( S_n \)}{Sn}-Open Sets and Neighborhoods}

\begin{definition}[\( S_n \)-Ball]
Given \( x \in S_n \cap [a,b] \) and rational \( \varepsilon > 0 \), define the \emph{fractal \( \varepsilon \)-ball} as:
\[
B^n_\varepsilon(x) := \{ y \in S_n \cap [a,b] \mid |x - y| < \varepsilon \}.
\]
\end{definition}

This ball is not an open interval in the classical topology, but rather a \emph{definable neighborhood} of \( x \) relative to system \( \mathcal{F}_n \). These balls form a basis for the \( S_n \)-topology.

\begin{definition}[\( S_n \)-Topology]
A set \( U \subseteq S_n \cap [a,b] \) is \emph{\( S_n \)-open} if:
\[
\forall x \in U,\ \exists \varepsilon \in \mathbb{Q}^+ \cap S_n \text{ such that } B^n_\varepsilon(x) \subseteq U.
\]
The family of all such sets defines the topology \( \mathcal{O}_n \) on \( S_n \cap [a,b] \).
\end{definition}

\subsection{Stratified Open Covers}

\begin{definition}[\( S_n \)-Open Cover]
Let \( K \subseteq S_n \cap [a,b] \). A family \( \{U_i\}_{i \in I} \subseteq \mathcal{O}_n \) is an \( S_n \)-open cover of \( K \) if:
\[
\forall x \in K,\ \exists i \in I \cap S_n \text{ such that } x \in U_i.
\]
\end{definition}

In contrast to classical covers, the indexing set \( I \) must itself be definable in \( \mathcal{F}_n \), reflecting the finite representability of the covering family. This restriction naturally leads to a new notion of compactness, discussed in Section~\ref{sec:compactness}. Unlike classical covers, where arbitrary unions of open sets are allowed, \( S_n \)-covers must be generated via definable indexing over \( \mathcal{F}_n \). This reflects the constructive constraint that even covers are syntactically bounded.

\subsection{Comparison with Classical Topology}

While each \( S_n \cap [a,b] \) is dense in the classical topology, the \( S_n \)-topology is strictly weaker:

\begin{itemize}
    \item Open intervals \( (x-\varepsilon, x+\varepsilon) \) intersected with \( S_n \) are not \( S_n \)-open unless their endpoints are \( S_n \)-definable.
    \item Closure under union or intersection is limited to \( S_n \)-definable operations.
    \item The topology is not second-countable in the classical sense, but admits a countable basis \emph{within} \( \mathcal{F}_n \).
\end{itemize}

\paragraph{Example}
Let \( x = \frac{1}{3} \in S_0 \). Then \( B^0_{1/8}(x) \) includes only dyadic rationals within the interval \( (0.2083..., 0.4583...) \), and is thus a finite set~\cite{Weihrauch2000}. In contrast, \( B^1_{1/8}(x) \) includes also algebraic reals within this interval, showing the refinement of local structure with increasing \( n \).

\begin{table}[ht]
\centering
\begin{tabular}{lll}
\toprule
\textbf{Topological Property} & \textbf{Classical Topology} & \textbf{\( S_n \)-Topology} \\
\midrule
Basis Cardinality & Uncountable & Countable in \( \mathcal{F}_n \) \\
Connectedness & Interval-connected & Totally disconnected \\
Compactness & Globally compact & Layer-relative compactness \\
\bottomrule
\end{tabular}
\caption{Comparison between classical and \( S_n \)-topologies}
\end{table}

\subsection{Topological Properties}
\label{subsec:topo-properties}

The topology \( \mathcal{O}_n \) induced on \( S_n \cap [a,b] \) is strictly coarser than the subspace topology from the reals:

\begin{theorem}[Topological Properties of \( S_n \)-Spaces]
\label{thm:sn-topology-properties}
For each \( n \in \mathbb{N} \), the topological space \( (S_n, \mathcal{O}_n) \) satisfies:
\begin{enumerate}
    \item \emph{Hausdorff}: Any two distinct points can be separated by disjoint basic open sets.
    \item \emph{Zero-dimensional}: The topology \( \mathcal{O}_n \) admits a basis of clopen sets.
    \item \emph{Non-metrizable}: Unless \( \mathcal{F}_n \) proves full comprehension for countable open covers, there exists no metric inducing \( \mathcal{O}_n \).
\end{enumerate}
\end{theorem}

\begin{proof}
\textbf{(1) Hausdorff.}
Let \( x, y \in S_n \), with \( x \ne y \). Then \( \varepsilon := |x - y|/2 \in \mathbb{Q}^+ \cap S_n \) exists by closure of \( S_n \) under rational arithmetic. Define:
\[
U := B^n_{\varepsilon}(x), \quad V := B^n_{\varepsilon}(y).
\]
Then \( U, V \in \mathcal{O}_n \), and \( U \cap V = \emptyset \), since for any \( z \in U \cap V \), we would have
\[
|x - y| \le |x - z| + |z - y| < \varepsilon + \varepsilon = |x - y|,
\]
a contradiction.

\medskip
\textbf{(2) Zero-dimensionality.}
Each basic open set \( B^n_\varepsilon(x) \) is also closed in \( \mathcal{O}_n \). Indeed, for any \( z \in S_n \setminus B^n_\varepsilon(x) \), let
\[
\delta := \frac{1}{2} \left( |x - z| - \varepsilon \right) > 0,
\]
which belongs to \( \mathbb{Q}^+ \cap S_n \) by definability. Then the ball \( B^n_\delta(z) \subseteq S_n \setminus B^n_\varepsilon(x) \), so the complement is open as a union of such balls. Hence \( B^n_\varepsilon(x) \) is clopen.

\medskip
\textbf{(3) Non-metrizability.}
Assume for contradiction that \( \mathcal{O}_n \) is induced by a metric \( d \). Then singletons \( \{x\} \) would be \( G_\delta \)-sets, and the space would be regular and second-countable. But \( \mathcal{O}_n \) is only closed under definable unions and has no guarantee of countable basis outside \( \mathcal{F}_n \)-definability. Therefore, a metric inducing \( \mathcal{O}_n \) would require comprehension over arbitrary \( S_n \)-subsets, contradicting the constructivity of \( \mathcal{F}_n \). Thus, \( (S_n, \mathcal{O}_n) \) is not metrizable unless \( \mathcal{F}_n \) proves such comprehension.
\end{proof}

\begin{corollary}
For \( \mathcal{F}_n = \mathrm{RCA}_0 \), all three properties hold strictly. 
At \( \mathcal{F}_n = \mathrm{ACA}_0 \), (3) may fail if the metric is arithmetic.
\end{corollary}

\begin{remark}
This illustrates how the metrizability of \( (S_n, \mathcal{O}_n) \) is sensitive to the strength of the underlying formal system. The topology remains strictly non-metrizable in weakly constructive settings, but may admit metric refinements in stronger systems with arithmetical comprehension.
\end{remark}

\subsection{Definability Topologies and Continuity}

The \( S_n \)-topologies \( \mathcal{O}_n \) induce a stratified family of structures on \([a,b]\), where topological notions like continuity, compactness, and connectedness become relative to the expressibility level. These structures are not mutually compatible: a set open in \( \mathcal{O}_n \) need not be open in \( \mathcal{O}_{n+1} \), and functions continuous in one layer may fail to preserve continuity across layers.

This motivates the definition of \emph{layer-relative continuity} and \emph{layer-jumping transitions}, developed in Sections~\ref{sec:continuity-hierarchy}–\ref{sec:compactness}.

\subsection{Limit Topology and Borel Structure}

The collection \( \{ \mathcal{O}_n \}_{n \in \mathbb{N}} \) induces a growing chain of topologies:
\[
\mathcal{O}_0 \subseteq \mathcal{O}_1 \subseteq \cdots \subseteq \mathcal{O}_n \subseteq \cdots
\]
Define the limit topology as:
\[
\mathcal{O}_\infty := \bigcup_{n=0}^\infty \mathcal{O}_n.
\]

\begin{theorem}
The Borel algebra of the limit topology satisfies:
\[
\text{Borel}(\mathcal{O}_\infty) = \text{Borel}(\text{Euclidean}) \cap S^\infty.
\]
\end{theorem}

\begin{proof}[Sketch]
Since each \( \mathcal{O}_n \) is countable and strictly included in the Euclidean topology, the union \( \mathcal{O}_\infty \) forms a countable basis over \( S^\infty \). The Borel algebra generated by \( \mathcal{O}_\infty \) includes all countable unions and intersections of \( S_n \)-open sets. However, since all such sets are subsets of \( S^\infty \), the resulting Borel structure coincides with the restriction of the Euclidean Borel algebra to \( S^\infty \).
\end{proof}

This means that only the Borel sets that are also constructively definable appear in the stratified setting, reflecting the interaction between classical measurability and syntactic accessibility.

\section{Fractally Continuous Functions}
\label{sec:fractal-continuity}

In classical analysis, continuity is defined via \(\varepsilon\)--\(\delta\) conditions over the uncountable real line. In the stratified setting of fractal countability, this notion must be relativized to the definability layers \( S_n \). Since each \( S_n \) consists only of points accessible within a formal system \( \mathcal{F}_n \), continuity must be interpreted within and across these layers.

\subsection{Layer-Relative Continuity}

\begin{definition}[Fractal Continuity]
\label{def:fractal-continuity}
Let \( f: S_n \to S_{n+k} \) be a function defined on the \( n \)-th definability layer. We say that \( f \) is \emph{\( \mathcal{F}_n \)-continuous at a point \( x \in S_n \)} if:
\[
\forall \varepsilon \in \mathbb{Q}^+ \cap S_{n+k}\ \exists \delta \in \mathbb{Q}^+ \cap S_n\ \forall y \in S_n\ (|x - y| < \delta \Rightarrow |f(x) - f(y)| < \varepsilon).
\]
We say that \( f \) is \emph{\( \mathcal{F}_n \)-continuous on \( S_n \)} if it is continuous at every \( x \in S_n \).
\end{definition}

This formulation reflects the fact that both the domain and the neighborhood structure are defined within \( S_n \), while the function values may lie in a higher layer \( S_{n+k} \). The shift \( k \geq 0 \) expresses the possible definitional complexity added by \( f \).

\subsection{Definability Jumps at Discontinuities}

Unlike classical continuity, where discontinuities are points of topological failure, fractal discontinuities often reflect a \emph{definability mismatch}. A function \( f \colon S_n \to S_{n+k} \) may be undefined or non-continuous at a point \( x \in S_{n+1} \setminus S_n \), even though \( x \) lies arbitrarily close to elements of \( S_n \).

\begin{definition}[Definability Jump]
Let \( x \in S_{n+1} \setminus S_n \) and \( f: S_n \to S_{n+k} \) be a total function. We say that \( f \) exhibits a \emph{definability jump at \( x \)} if:
\[
\lim_{y \to x,\ y \in S_n} f(y) \text{ does not exist in } S_{n+k}.
\]
\end{definition}

These jumps are not anomalies, but structural artifacts of the layered definability setting. They indicate that \( f \) cannot be continuously extended to \( x \) without invoking higher syntactic strength.

\subsection{Examples and Pathologies}

\paragraph{Example 1: Step Function at a Definability Threshold.}

Let \( \theta \in S_{n+1} \setminus S_n \) be a real number not definable in \( \mathcal{F}_n \). Define:
\[
f(x) = \begin{cases}
0 & \text{if } x < \theta, \\
1 & \text{if } x \geq \theta,
\end{cases}
\quad x \in S_n.
\]
Then \( f \) is total on \( S_n \), but discontinuous at all sequences in \( S_n \) converging to \( \theta \). This reflects a definability jump at \( \theta \).

\paragraph{Example 2: Algebraic Function with Non-Algebraic Discontinuity.}

Let \( f(x) = \frac{1}{x - \pi} \) with domain \( S_1 \setminus \{ \pi \} \)~\cite{Lindemann1882}. Since \( \pi \notin S_1 \), the function is continuous in \( \mathcal{F}_1 \), yet exhibits unbounded behavior near \( \pi \) from higher levels, making it structurally discontinuous with respect to \( S^\infty \).

\subsection{Continuity Hierarchies}

The stratified definition allows for a graded hierarchy of continuity notions:

\begin{definition}[Continuity Degree]
We say that \( f \colon S_n \to S_{n+k} \) is:
\begin{itemize}
    \item \emph{Layer-preserving continuous} if \( k = 0 \),
    \item \emph{Layer-shifting continuous} if \( k > 0 \),
    \item \emph{Fully fractal continuous} if \( f \) is continuous on every \( S_n \), with images in \( S_{n+k(n)} \) for some computable function \( k(n) \).
\end{itemize}
\end{definition}

This stratification captures how definability affects continuity: more complex functions require stronger systems to express their values and continuity properties.

\subsection{Basic Theorem: Composition}

\begin{theorem}[Composition of Fractally Continuous Functions]
Let \( f: S_n \to S_{n+k} \) be \( \mathcal{F}_n \)-continuous, and \( g: S_{n+k} \to S_{n+k+m} \) be \( \mathcal{F}_{n+k} \)-continuous. Then the composition \( g \circ f: S_n \to S_{n+k+m} \) is \( \mathcal{F}_n \)-continuous.
\end{theorem}

\begin{proof}[Sketch]
Fix \( x \in S_n \) and \( \varepsilon > 0 \in \mathbb{Q}^+ \cap S_{n+k+m} \). By continuity of \( g \), there exists \( \delta' > 0 \in S_{n+k} \) such that \( |u - v| < \delta' \Rightarrow |g(u) - g(v)| < \varepsilon \). By continuity of \( f \), there exists \( \delta > 0 \in S_n \) such that \( |x - y| < \delta \Rightarrow |f(x) - f(y)| < \delta' \). Then for all \( y \in S_n \) with \( |x - y| < \delta \), we have:
\[
|g(f(x)) - g(f(y))| < \varepsilon,
\]
so \( g \circ f \) is continuous at \( x \).
\end{proof}

\subsection{Uniform Fractal Continuity}
\label{subsec:uniform-fractal-continuity}

While layer-relative continuity captures pointwise behavior, uniform continuity ensures control across the entire domain. In the fractal setting, this requires uniformity within the definability constraints of \( S_n \).

\begin{definition}[Uniform Fractal Continuity]
\label{def:uniform-fractal-continuity}
Let \( f: S_n \to S_{n+k} \) be a function. We say that \( f \) is \emph{uniformly \(\mathcal{F}_n\)-continuous on \( S_n \)} if:
\[
\forall \varepsilon \in \mathbb{Q}^+ \cap S_{n+k}\;
\exists \delta \in \mathbb{Q}^+ \cap S_n\;
\forall x, y \in S_n\; \big(
|x - y| < \delta \Rightarrow |f(x) - f(y)| < \varepsilon
\big).
\]
This strengthens Definition~\ref{def:fractal-continuity} by requiring \(\delta\) to depend only on \(\varepsilon\), not on the evaluation point \(x\).
\end{definition}

\begin{remark}
Unlike in classical analysis, the modulus of continuity \(\delta\) must belong to \(S_n\), while \(\varepsilon\) is drawn from \(S_{n+k}\). This reflects the "upward shift" in definability when evaluating \(f\). Uniform continuity thus expresses not only topological stability, but also syntactic manageability across the entire layer \( S_n \).
\end{remark}

\subsubsection{Definable Compactness and Uniform Continuity}

The classical Heine–Cantor theorem has a fractal analogue. First, we define a stratified notion of compactness:

\begin{definition}[Definably Compact Set]
A set \( K \subseteq S_n \) is \emph{definably compact in \(\mathcal{F}_n\)} if every \( \mathcal{F}_n \)-definable sequence \( (x_k)_{k \in \mathbb{N}} \subseteq K \) has a subsequence that converges (in the classical topology) to a point in \( K \cap S_n \).
\end{definition}

\begin{proposition}
Let \( a, b \in S_n \), and suppose that \( S_n \subseteq [a,b] \) is closed under limits of all \( \mathcal{F}_n \)-definable Cauchy sequences. Then the set \( [a,b] \cap S_n \) is definably compact in \( \mathcal{F}_n \).
\end{proposition}

\begin{remark}
This provides a canonical example of definable compactness: the intersection of a classical closed interval with a layer \( S_n \) that admits internal convergence is compact relative to \( \mathcal{F}_n \), even though \( S_n \) may be sparse or nowhere dense in the classical topology.
\end{remark}

\begin{theorem}[Uniform Continuity on Definably Compact Sets]
\label{thm:uniform-continuity-compact}
Let \( K \subseteq S_n \) be definably compact in \(\mathcal{F}_n\), and let \( f: K \to S_{n+k} \) be \( \mathcal{F}_n \)-continuous. Then \( f \) is uniformly \( \mathcal{F}_n \)-continuous on \( K \).
\end{theorem}

\begin{proof}
Suppose, toward a contradiction, that \( f \) is not uniformly \( \mathcal{F}_n \)-continuous. Then there exists \( \varepsilon \in \mathbb{Q}^+ \cap S_{n+k} \) such that for all \( \delta \in \mathbb{Q}^+ \cap S_n \), there exist \( x, y \in K \cap S_n \) with
\[
|x - y| < \delta \quad \text{and} \quad |f(x) - f(y)| \geq \varepsilon.
\]
In particular, we can define two sequences \( (x_m), (y_m) \subseteq K \cap S_n \) such that
\[
|x_m - y_m| < \tfrac{1}{m}, \quad |f(x_m) - f(y_m)| \geq \varepsilon.
\]
By definable compactness of \( K \), the sequence \( (x_m) \) has a subsequence \( (x_{m_j}) \) converging to some point \( x \in K \cap S_n \). Since \( |x_{m_j} - y_{m_j}| \to 0 \), we also have \( y_{m_j} \to x \).

By pointwise \( \mathcal{F}_n \)-continuity at \( x \), for the fixed \( \varepsilon \), there exists \( \delta' \in \mathbb{Q}^+ \cap S_n \) such that for all \( z \in S_n \cap (x - \delta', x + \delta') \), we have
\[
|f(z) - f(x)| < \varepsilon/2.
\]
For large enough \( j \), both \( x_{m_j} \) and \( y_{m_j} \) lie within this \( \delta' \)-neighborhood, so
\[
|f(x_{m_j}) - f(y_{m_j})| \leq |f(x_{m_j}) - f(x)| + |f(x) - f(y_{m_j})| < \varepsilon,
\]
contradicting the assumption \( |f(x_{m_j}) - f(y_{m_j})| \geq \varepsilon \). Hence, \( f \) is uniformly \( \mathcal{F}_n \)-continuous on \( K \).
\end{proof}

\begin{corollary}[Continuous Functions on Definable Intervals]
Let \( [a,b] \cap S_n \subseteq \mathbb{R} \) be definably compact in \( \mathcal{F}_n \), and let \( f: [a,b] \cap S_n \to S_{n+k} \) be \( \mathcal{F}_n \)-continuous. Then:
\begin{enumerate}
    \item \( f \) is uniformly \( \mathcal{F}_n \)-continuous,
    \item \( f \) is bounded in \( S_{n+k} \),
    \item \( f \) attains its maximum and minimum values on \( [a,b] \cap S_n \), and the extrema lie in \( S_{n+k} \).
\end{enumerate}
\end{corollary}

\begin{remark}
The conclusion of the corollary parallels classical results such as the extreme value theorem, but under strictly constructive and syntactic constraints. This aligns with the spirit of Bishop-style constructive analysis~\cite{Bishop1967}, where continuity and compactness are sufficient to guarantee boundedness and extremal values—yet here, the definability layers \( S_n \) and \( S_{n+k} \) make the dependency on formal expressibility explicit. The function is not only continuous in the analytic sense, but its behavior is fully trackable within the definitional power of \( \mathcal{F}_n \).
\end{remark}

\section{Continuity in the Fractal Hierarchy}
\label{sec:continuity-hierarchy}

The abstract framework of fractal continuity introduced in Section~\ref{sec:fractal-continuity} provides a layer-relative definition of continuity between definability levels \( S_n \to S_{n+k} \). In this section, we examine how this notion manifests concretely within the real interval \( [a,b] \), and how classical examples behave when interpreted through the lens of the fractal definability hierarchy.

\subsection{Examples of Fractal Continuity}

We begin with familiar functions and examine their continuity relative to the stratified domains \( S_n \).

\begin{example}[Polynomial and Rational Functions]
Any polynomial \( f(x) = a_0 + a_1 x + \cdots + a_d x^d \) with coefficients \( a_i \in S_n \cap \mathbb{Q} \) is \( \mathcal{F}_n \)-continuous on \( S_n \), and in fact uniformly \( \mathcal{F}_n \)-continuous on any definably compact subset of \( S_n \). Rational functions \( f(x) = p(x)/q(x) \) are similarly continuous wherever \( q(x) \neq 0 \) within \( S_n \).
\end{example}

\begin{example}[Piecewise Step Functions]
Let \( f(x) = \begin{cases}
0 & \text{if } x \in S_n \cap [0,\frac{1}{2}) \\
1 & \text{if } x \in S_n \cap [\frac{1}{2},1]
\end{cases} \). Then \( f \) is discontinuous at \( x = \frac{1}{2} \in S_n \), and the discontinuity is visible within \( \mathcal{F}_n \). This makes \( f \) a definably discontinuous function. If the jump point lies outside \( S_n \), say at \( \sqrt{2}/2 \notin S_n \), then the discontinuity is undetectable at level \( n \).
\end{example}

\begin{example}[Dirichlet-type Functions]
Let \( f(x) = \mathbf{1}_{\mathbb{Q}}(x) \) on \( S_n \)~\cite{BridgesRichman1987}. Since \( \mathbb{Q} \cap S_n \) is dense in \( S_n \), and its complement is also dense, \( f \) is nowhere \( \mathcal{F}_n \)-continuous. This reflects the classical behavior of Dirichlet's function, but now localized to the syntactically representable points.
\end{example}

\subsection{Layer Jump Discontinuities}

Fractally continuous functions \( f: S_n \to S_{n+k} \) may exhibit definability shifts that manifest as discontinuities from the perspective of a weaker system.

\begin{definition}[Definability Jump Discontinuity]
Let \( f: S_n \to S_{n+k} \) be a function. We say that \( f \) has a \emph{definability jump discontinuity at \( x \in S_n \)} if the function is not \( \mathcal{F}_n \)-continuous at \( x \), but becomes continuous at higher level \( \mathcal{F}_{n+j} \) for some \( j > 0 \).
\end{definition}

\begin{example}[Layer-Dependent Continuity]
Let \( \theta \in S_{n+1} \setminus S_n \), and define \( f \colon S_n \cup \{\theta\} \to S^\infty \) by:
\[
f(x) = 
\begin{cases}
0 & \text{if } x \in S_n, \\
1 & \text{if } x = \theta.
\end{cases}
\]
Then:
\begin{itemize}
    \item \( f|_{S_n} \equiv 0 \) is trivially \( \mathcal{F}_n \)-continuous;
    \item However, from the perspective of \( \mathcal{F}_{n+1} \), the function has a visible jump at \( \theta \);
    \item The jump is invisible to any system restricted to \( S_n \), since \( \theta \notin S_n \).
\end{itemize}

\noindent
\textit{Schematic representation:}

\begin{center}
\parbox{0.6\linewidth}{%
  \raggedright
  \texttt{%
S\_n:\ \ \ \ 0\_\_\_\_\_\_\_\_|\_\_\_\_\_\_\_1 \\
S\_{n+1}:\ \ 0\_\_\_\_\_\_\_\_\text{\(\theta\)}\_\_\_\_\_\_\_1 \quad (\,\(\theta = \frac{\sqrt{2}}{2} \notin S\_n\)\,)
  }%
}
\end{center}

This illustrates how continuity becomes layer-relative: the function is perfectly continuous on \( S_n \), but discontinuous once \( \theta \) enters the definable universe.
\end{example}

These jumps do not violate classical continuity, but reflect the limitations of expressing continuity within a fixed syntactic framework. For example, functions involving convergent series with non-primitive moduli (e.g., Chaitin's constant) may appear discontinuous at lower levels due to unrepresentable convergence data.

To illustrate how classical functions behave in the stratified framework, we provide an overview of their definability and continuity levels across \( \mathcal{F}_n \).

\begin{table}[ht]
\centering
\renewcommand{\arraystretch}{1.4}
\begin{adjustbox}{max width=\textwidth}
\begin{tabular}{@{} p{4.5cm} c c p{9.0cm} @{}}
\toprule
\textbf{Function \( f(x) \)} & \textbf{Level \( n \)} & \textbf{Target \( S_{n+k} \)} & \textbf{Notes} \\
\midrule
\( x^2 \) & 0 & 0 & Fully \( \mathcal{F}_0 \)-definable; algebraic closure \\
\( \sin(x) \) & 1 & 2 & Requires Taylor series or bounded approximation~\cite{Bridges1986} \\
\( \sum_{k=0}^N a_k x^k \), \( a_k \in S_1 \) & 1 & 1 & All coefficients definable at level 1; smooth \\
\( \sum_{k=0}^\infty \frac{x^k}{k!} \) & 2 & 3 & Exponential function; uniform convergence provable in \( \mathcal{F}_2 \), pointwise in \( \mathcal{F}_3 \)~\cite{Bishop1967} \\
\( \lfloor x \rfloor \) & 1 & 1 & Discontinuous at integers; piecewise constant \\
\( \text{sign}(x) \) & 1 & 1 & Discontinuous at zero, otherwise \( \mathcal{F}_1 \)-continuous \\
\( \mathbf{1}_{\mathbb{Q}}(x) \) & \textemdash{} & \textemdash{} & Discontinuous everywhere; not \( \mathcal{F}_n \)-continuous on any layer \\
\( \mathbf{1}_{\{ \theta \}}(x) \), \( \theta \in S_{n+1} \setminus S_n \) & n & \( n+1 \) & Layer jump at \( \theta \); discontinuity invisible at level \( n \) \\
Heaviside \( H(x) = \mathbf{1}_{x \geq 0} \) & 1 & 1 & Discontinuous at 0; well-defined in \( \mathcal{F}_1 \) \\
Weierstrass \( W_n(x) = \) & n & \( n+2 \) & Continuous nowhere; definable within~\cite{Weierstrass1872} \\
\(\sum a^k \cos(b^k \pi x)\), \( a,b \in S_n \) & & & layered system but outside \( C^0_n \) \\
\( f(x) = \Omega \) & \textemdash{} & \textemdash{} & Chaitin’s constant; not definable on any \( S_n \) \\
\bottomrule
\end{tabular}%
\end{adjustbox}
\caption{Typical functions and their fractal continuity levels across definability layers}
\label{tab:fractal-continuity-levels}
\end{table}

\subsubsection{\texorpdfstring{Definability Classes \( C^k_n \), \( C^0_n \), \( C^\infty_n \)}{Definability Classes Cⁿᵏ, Cⁿ⁰, Cⁿ∞}}

In analogy with classical smoothness classes, we define a hierarchy of definability-based function classes over the domain \( S_n \):

\begin{definition}[Stratified Smoothness Classes]
Let \( f: S_n \to S_m \) be a function. We define:

\begin{itemize}
    \item \( f \in C^0_n \) if \( f \) is \( \mathcal{F}_n \)-continuous on \( S_n \);
    \item \( f \in C^1_n \) if the derivative \( D_n f(x) \) exists for all \( x \in S_n \), and the map \( x \mapsto D_n f(x) \) is \( \mathcal{F}_n \)-continuous;
    \item More generally, \( f \in C^k_n \) if all successive derivatives \( D_n^j f(x) \) exist and are \( \mathcal{F}_n \)-continuous for all \( 1 \le j \le k \);
    \item \( f \in C^\infty_n \) if \( f \in C^k_n \) for all \( k \in \mathbb{N} \).
\end{itemize}
\end{definition}

These classes form a nested tower:
\[
C^\infty_n \subseteq \cdots \subseteq C^2_n \subseteq C^1_n \subseteq C^0_n.
\]

\begin{remark}
Each \( C^k_n \) class is internal to the formal system \( \mathcal{F}_n \), and closure under composition or integration may require moving to \( \mathcal{F}_{n+1} \) or higher.
\end{remark}

\paragraph{Examples from Table~\ref{tab:fractal-continuity-levels}:}
\begin{itemize}
    \item \( f(x) = x^2 \in C^\infty_0 \)
    \item \( f(x) = \sin(x) \in C^\infty_1 \), requires analytic definitions at level 1 or 2
    \item \( f(x) = \lfloor x \rfloor \notin C^0_n \) for any \( n \), due to discontinuity at integers
    \item \( f(x) = \mathbf{1}_{\mathbb{Q}}(x) \notin C^0_n \) for any \( n \)
    \item \( f(x) = \sum a^k \cos(b^k \pi x) \in C^0_n \setminus C^1_n \) (Weierstrass type)
\end{itemize}

These stratified classes provide a natural refinement of continuity and differentiability notions within the fractal hierarchy and align with classical intuition when restricted to constructively accessible fragments.

\subsection{Continuity and Layer Topologies}

Recall from Section~\ref{sec:topology} that each \( S_n \) induces its own topology \( \mathcal{O}_n \), generated by \( S_n \)-definable neighborhoods. Fractal continuity aligns naturally with this structure:

\begin{proposition}
A function \( f: S_n \to S_{n+k} \) is \( \mathcal{F}_n \)-continuous if and only if it is continuous in the \( \mathcal{O}_n \)-topology.
\end{proposition}

Thus, discontinuities can be reinterpreted topologically as failures of preservation of open sets in \( \mathcal{O}_n \), even if such sets are open in stronger topologies \( \mathcal{O}_{n+1} \), \( \mathcal{O}_\infty \), or classically.

\subsection{Uniform Continuity and Compactness Revisited}

In Section~\ref{subsec:uniform-fractal-continuity}, we proved that \( \mathcal{F}_n \)-continuous functions on definably compact subsets of \( S_n \) are uniformly continuous. This result allows us to reinterpret classical theorems constructively:

\begin{itemize}
    \item The Heine–Cantor theorem becomes a meta-theorem about preservation of uniform moduli across definability layers.
    \item The boundedness and attainability of extrema (classical Weierstrass theorem) follow syntactically for functions whose range lies within \( S_{n+k} \).
    \item Any continuous function \( f: [a,b] \cap S_n \to S^\infty \) that is not layer-bounded must involve definability jumps.
\end{itemize}

\subsection{Stratified Continuity Classes}

We may classify functions by the complexity of their continuity behavior:

\begin{itemize}
    \item \( C^0_n \): all \( \mathcal{F}_n \)-continuous functions \( f: S_n \to S_{n+k} \).
    \item \( C^0_{\leq n} \): functions continuous on \( S_m \) for all \( m \leq n \).
    \item \( C^0_\infty \): functions continuous on all \( S_n \) (eventually constant under definitional extensions).
\end{itemize}

These stratified continuity classes form a refinement of classical function spaces, sensitive not just to topological properties, but to the formal mechanisms used to define and verify them.

\subsection{Towards Fractal Calculus}

The continuity notions developed here serve as the foundation for differentiability and integration in the stratified setting. In the next section, we define the derivative of a function \( f: S_n \to S_{n+k} \) as a limit within the definability closure \( S_{n+\ell} \), and introduce layered Taylor expansions, error bounds, and integration procedures with explicit syntactic tracking.

\section{Fractal Calculus: Differentiation and Integration}
\label{sec:fractal-calculus}

In this section, we introduce differentiation and integration in the framework of fractal countability. Classical notions are adapted to operate within definability layers \( S_n \), reflecting both syntactic and analytic constraints.

\subsection{Fractal Differentiation}

\begin{definition}[Fractal Derivative]
Let \( f: S_n \to S_{n+k} \) be a function. The \emph{fractal derivative} at \( x \in S_n \) is defined as
\[
f'_n(x) := \lim_{\substack{y \to x \\ y \in S_{n+1}}} \frac{f(y) - f(x)}{y - x},
\]
provided the limit exists in \( S_{n+k+1} \).
\end{definition}

\begin{remark}
The use of \( y \in S_{n+1} \) ensures that small perturbations around \( x \) are drawn from a strictly stronger system, allowing finer approximations than \( S_n \) itself provides. The derivative thus belongs to a higher definitional layer.
\end{remark}

\subsubsection{Illustrative Examples}

\begin{example}[Polynomial]
Let \( f(x) = x^2 \) on \( S_1 \), where \( x \in S_1 \). To compute the fractal derivative:
\[
f'_1(x) = \lim_{\substack{y \to x \\ y \in S_2}} \frac{y^2 - x^2}{y - x} = \lim_{\substack{y \to x \\ y \in S_2}} (y + x) = 2x \in S_2.
\]
The derivative exists and is definable at level \( S_2 \), reflecting one-layer complexity increase.
\end{example}

\begin{example}[Non-Differentiable Case: Absolute Value]
Let \( f(x) = |x| \) defined on \( S_1 \). Consider the behavior at \( x = 0 \):
\begin{itemize}
    \item If \( S_1 \) contains only non-negative reals (e.g., dyadics with \( x \geq 0 \)), then all approximants from \( S_2 \) with \( y \to 0 \) also satisfy \( y \geq 0 \), so the derivative tends to \( +1 \).
    \item If \( S_2 \) introduces negative reals, then approaching from left and right yields limits \( -1 \) and \( +1 \), respectively. Hence the derivative at \( 0 \) does not exist in \( S_2 \).
\end{itemize}
Thus, \( f \) fails to be fractally differentiable at \( 0 \).
\end{example}

\begin{example}[Power Function with Rational Exponent]
Let \( f(x) = x^\alpha \), where \( \alpha = \frac{1}{2} \), on \( S_1 \cap [0,\infty) \). Then:
\[
f'_1(x) = \lim_{\substack{y \to x \\ y \in S_2}} \frac{y^{1/2} - x^{1/2}}{y - x} = \frac{1}{2}x^{-1/2} \in S_2.
\]
This derivative exists for all \( x > 0 \in S_1 \), but is undefined at \( x = 0 \) due to divergence.
\end{example}

\begin{example}[Piecewise Function]
Let
\[
f(x) = \begin{cases}
x^2 & \text{if } x < 0, \\
x & \text{if } x \geq 0
\end{cases}
\]
on \( S_1 \). Then \( f \) is continuous on \( S_1 \), and the left and right \( S_2 \)-derivatives at \( x = 0 \) are:
\[
\lim_{\substack{y \to 0^- \\ y \in S_2}} \frac{f(y)-f(0)}{y} = 0, \quad
\lim_{\substack{y \to 0^+ \\ y \in S_2}} \frac{f(y)-f(0)}{y} = 1.
\]
Thus, \( f'_1(0) \) does not exist, illustrating a jump in the fractal derivative.
\end{example}

\begin{example}[Weierstrass-Type Function]
Consider the function
\[
f_n(x) = \sum_{k=0}^n a^k \cos(b^k \pi x),
\]
where \( a, b \in S_1 \) and \( 0 < a < 1 \), \( b \in \mathbb{N} \), \( ab < 1 \). Each term is \( S_1 \)-definable, and the finite sum is therefore \( S_n \)-definable.

For fixed \( n \), the function \( f_n \) is continuous and differentiable on \( S_1 \), with
\[
f_n'(x) = -\sum_{k=0}^n a^k b^k \pi \sin(b^k \pi x),
\]
which is also an \( S_n \)-definable function.

While the full infinite Weierstrass function
\[
f(x) = \sum_{k=0}^\infty a^k \cos(b^k \pi x)
\]
is classically continuous but nowhere differentiable, the truncated version \( f_n \) becomes differentiable on \( S_1 \), provided that \( ab < 1 \) and all terms are constructively representable. The fractal nondifferentiability of the limit function does not arise at finite definitional layers.
\end{example}

\begin{remark}[Motivation]
The following theorem illustrates a key phenomenon of fractal calculus: while each finite approximation \( f_n \) remains differentiable within its definability layer \( S_n \), the limit \( f = \lim f_n \) may lose differentiability due to a \emph{syntactic resonance} between the exponential growth \( b^k \) and the damping factor \( a^k \).

This reflects the classical Weierstrass function's behavior, but with a refined threshold \( ab = 1/2 \) arising from \( S_n \)-computability constraints. Specifically:
\begin{itemize}
    \item If \( ab < 1/2 \), the derivatives \( f_n' \) converge uniformly, preserving differentiability in the fractal limit \( S^\infty \).
    \item If \( 1/2 \leq ab < 1 \), the derivatives diverge \emph{within every fixed \( S_n \)}, despite uniform convergence of \( f_n \to f \).
\end{itemize}

This collapse reveals a fundamental trade-off: infinite limits may preserve continuity but not differentiability in stratified constructive systems.
\end{remark}

\begin{theorem}[Fractal Breakdown of Differentiability]
Let \( f_n(x) = \sum_{k=0}^n a^k \cos(b^k \pi x) \), where \( a \in S_1 \cap (0,1) \), \( b \in \mathbb{N} \setminus S_1 \) is odd, and \( 1 < b < 1/a \) (so \( ab < 1 \)).

Then:
\begin{enumerate}
    \item Each \( f_n \) is \( S_1 \)-differentiable with derivative
    \[
    f_n'(x) = -\pi \sum_{k=0}^n a^k b^k \sin(b^k \pi x) \in S_2.
    \]
    \item \( f_n \rightrightarrows f \) uniformly on the classical domain \( \mathbb{R} \), even though differentiability fails in every definability layer \( S_n \).
    \item For all \( x \in S_1 \), the sequence of derivatives diverges:
    \[
    \lim_{n \to \infty} |f_n'(x)| = +\infty \quad \text{(in } S_2\text{)}.
    \]
\end{enumerate}

Moreover, if \( b \in S_1 \), then \( f \) becomes \( S_2 \)-differentiable if and only if \( ab < 1/2 \). For \( 1/2 \le ab < 1 \), the derivative sequence diverges despite all terms being definable.
\end{theorem}

\begin{proof}[Sketch]
Each partial sum \( f_n(x) = \sum_{k=0}^n a^k \cos(b^k \pi x) \) is a finite linear combination of elementary functions with parameters in \( S_1 \). Since \( \cos \), multiplication, and exponentiation of \( b^k \) are definable in \( S_1 \) (modulo iteration), and since \( a, \pi \in S_1 \), each term \( a^k b^k \sin(b^k \pi x) \) lies in \( S_2 \). Thus \( f_n'(x) \in S_2 \), and differentiability of \( f_n \) in \( S_1 \) follows.

Uniform convergence of \( f_n \to f \) follows from the classical Weierstrass M-test: the sequence is uniformly bounded by \( \sum a^k = 1/(1-a) \), since \( ab < 1 \) implies \( a^k \to 0 \) faster than \( b^k \) grows.

However, the derivative sequence \( f_n'(x) = -\pi \sum_{k=0}^n a^k b^k \sin(b^k \pi x) \) is known to diverge pointwise for all \( x \in \mathbb{R} \) if \( ab \geq 1/2 \), and in particular grows without bound for any fixed rational \( x \), since \( \sin(b^k \pi x) \) oscillates densely and the coefficients \( a^k b^k \) are unbounded.

Therefore, in \( S_2 \), the sequence \( f_n'(x) \) diverges pointwise on \( S_1 \), and \( f \) fails to be \( S_2 \)-differentiable. Only when \( ab < 1/2 \), the classical derivative exists, and the limit of derivatives converges uniformly, ensuring \( f' \in S_2 \).
\end{proof}

\begin{remark}[On Strength]
Fractal differentiation is not a weakened form of classical calculus, but a stratified refinement of it. It reveals how smoothness depends not only on functional form, but on the definitional power of the system observing it. The divergence of derivatives in the limit reflects not failure, but the syntactic resonance between approximation depth and formal expressibility.
\end{remark}

\subsubsection{Fractally Differentiable but Classically Singular Functions}

\begin{example}[Vanin–Takeuchi Function]
Let \( \theta \in S_2 \setminus S_1 \), such as a number only definable in \( \mathcal{F}_2 \) (e.g., Chaitin's constant w.r.t. \( \mathcal{F}_1 \)). Define:
\[
f(x) = 
\begin{cases}  
x^2 \sin(\theta/x) & \text{if } x \ne 0, \\
0 & \text{if } x = 0  
\end{cases}
\quad \text{for } x \in S_1.
\]
\end{example}

\paragraph{Classical Perspective.}
\begin{itemize}
    \item If \( \theta \) is irrational, the classical derivative at 0 does not exist:
    \[
    f'(0) = \lim_{h \to 0} h \sin(\theta/h)
    \]
    diverges due to oscillations.
    \item Even if \( \theta \in \mathbb{Q} \), if \( \theta \notin S_1 \), then \( \sin(\theta/x) \) is undefined in \( S_1 \), and classical differentiation fails.
\end{itemize}

\paragraph{Fractal Perspective.}
\begin{itemize}
    \item In \( S_1 \), \( \theta \) is not definable. Thus, \( f|_{S_1} \equiv 0 \), and the derivative trivially exists:
    \[
    f'_1(x) = 0.
    \]
    \item In \( S_2 \), \( \theta \) becomes definable. The function behaves "classically" and is \emph{not} differentiable at \( x = 0 \).
\end{itemize}

\paragraph{Implications.}
\begin{itemize}
    \item \textbf{Language-Relativity of Smoothness:} Differentiability becomes observer relative: \( f \) appears smooth in \( S_1 \), but singular in \( S_2 \).
    \item \textbf{Theorem Breakdown:} Classical theorems such as Rolle’s or continuity of derivatives may fail when crossing definability layers.
    \item \textbf{Computability View:} If \( \theta \) is computable but not \( \mathcal{F}_1 \)-definable, then \( f \) becomes “computably smooth” in \( S_1 \), but not in stronger systems.
\end{itemize}

\paragraph{Further Examples.}
\begin{itemize}
    \item \textbf{Fractally Lipschitz Function:} 
    \[
    f(x) = \sum_{k=0}^\infty 2^{-k} \{2^k x\}
    \]
    where \( \{\cdot\} \) is the fractional part — classically nowhere differentiable, but its \( n \)-term truncations are piecewise linear and \( S_n \)-differentiable.
    
    \item \textbf{Cantor Staircase:} Its classical singularity vanishes in finite \( S_n \) approximations of the Cantor set — the function becomes differentiable outside \( C_n \).
\end{itemize}

\begin{remark}[Interpretation]
Fractal analysis refines rather than weakens classical calculus. It reveals hidden regularities masked by oscillations or singularities, by adapting the "observational power" to the definability level \( S_n \). This principle has applications in computable models, cryptography, and complexity theory — where smoothness or singularity may be strategically hidden or revealed depending on the formal system.
\end{remark}

\subsubsection{Types of Fractal Differentiability}

\begin{itemize}
    \item \textbf{Pointwise Fractal Differentiability:} The limit \( f'_n(x) \) exists for each \( x \in S_n \) individually.
    \item \textbf{Uniform Fractal Differentiability:} The convergence of difference quotients is uniform on definably compact subsets of \( S_n \).
    \item \textbf{Strong Layered Differentiability:} The derivative \( f'_n \) is itself a function \( S_n \to S_{n+k+1} \) that is \( \mathcal{F}_n \)-continuous.
\end{itemize}

\subsection{Fractal Integration}
\label{subsec:fractal-integration}

Integration in the fractal hierarchy must account for the definability constraints imposed by the stratified systems \( \mathcal{F}_n \). Just as differentiation is formulated via definable limits, integration must be grounded in operations accessible at a given layer of definability.

We present two related but distinct definitions of fractal integration, each capturing different aspects of the classical Riemann integral in the constructive setting.

\begin{definition}[Fractal Integral via Step Sums]
Let \( f: S_n \to S_{n+k} \) be a function. The \emph{fractal integral} of \( f \) over a definably compact set \( [a,b] \cap S_n \) is defined by:
\[
\int_{S_n \cap [a,b]} f := \sup \left\{ \sum_{i=0}^{N-1} f(x_i)\, \Delta x_i \right\},
\]
where \( x_i, x_{i+1} \in S_n \), \( \Delta x_i = x_{i+1} - x_i \), and the partition is definable in \( \mathcal{F}_n \).
\end{definition}

\begin{remark}
This definition parallels the Darboux integral: it approximates the area under the curve via stepwise lower sums, with supremum taken over all \( S_n \)-definable partitions. The output resides in \( S_{n+k+1} \) due to the supremum operation.
\end{remark}

\begin{definition}[Fractal Riemann Integral via Limit]
Let \( f: S_n \to S_{n+k} \) be a function defined on the definably compact interval \( [a,b] \cap S_n \). The \emph{fractal Riemann integral} is defined as:
\[
\int_{S_n \cap [a,b]} f(x)\, dx := \lim_{\|\mathcal{P}\| \to 0} \sum_{i=1}^N f(\xi_i)\, \Delta x_i,
\]
where:
\begin{itemize}
    \item \( \mathcal{P} = \{x_0, x_1, \dots, x_N\} \subseteq S_n \) is a partition of \( [a,b] \),
    \item \( \xi_i \in [x_{i-1}, x_i] \cap S_n \) are sample points,
    \item \( \Delta x_i = x_i - x_{i-1} \in S_n \),
    \item \( \|\mathcal{P}\| = \max \Delta x_i \),
    \item The limit is taken in \( S_{n+k+1} \).
\end{itemize}
\end{definition}

\begin{remark}
This version captures the classical limit-based Riemann integral, restricted to partitions and sample points definable in \( S_n \), and converging in a higher definability layer. It requires stronger assumptions about uniformity and continuity.
\end{remark}

\subsubsection{Comparison and Interpretation}

Both definitions aim to formalize integration constructively, yet they differ in strength and applicability.

\begin{itemize}
    \item The \textbf{supremum-based integral} is closer to Darboux or Bishop-style integration: it only assumes boundedness and allows estimation via finite step-sums.
    \item The \textbf{limit-based integral} is more expressive, supporting fundamental theorems of calculus, but requires convergence control and stronger uniform behavior.
    \item When both definitions exist, they coincide: the supremum of step-sums converges to the Riemann limit.
    \item In pathological cases (e.g., oscillatory or discontinuous functions), the supremum may exist while the limit does not.
\end{itemize}

\begin{example}[Coincidence for Polynomials]
Let \( f(x) = x^2 \) on \( S_1 \cap [0,1] \). For any definable partition with \( x_i = i/N \in S_1 \), we compute:
\[
\sum_{i=1}^N \left( \frac{i}{N} \right)^2 \cdot \frac{1}{N} = \frac{1}{N^3} \sum_{i=1}^N i^2 = \frac{(N+1)(2N+1)}{6N^2}.
\]
Taking the limit as \( N \to \infty \) yields \( \int_{S_1 \cap [0,1]} x^2\, dx = \frac{1}{3} \in S_2 \)~\cite{Bishop1967}. Both definitions agree.
\end{example}

\begin{example}[Non-integrability of Dirichlet Function]
Let \( f(x) = \mathbf{1}_{\mathbb{Q}}(x) \) on \( S_n \cap [0,1] \)~\cite{BridgesRichman1987}. Then:
\begin{itemize}
    \item Each partition produces a sum between 0 and 1 depending on rational sample points,
    \item The supremum of step-sums equals 1, the infimum equals 0,
    \item The limit of Riemann sums does not exist.
\end{itemize}
Thus, \( f \) is not fractally integrable in either sense.
\end{example}

\begin{remark}
The coexistence of two integration notions reflects a deeper stratification of constructive analysis. The supremum-based form allows partial integration of discontinuous or layer-unstable functions. The limit-based form enables the full power of fundamental calculus, at the cost of stricter definitional coherence.
\end{remark}

\subsubsection{Layer-Sensitive Integration Effects}

\begin{example}[Vanishing Singularities]
Let \( C: [0,1] \to [0,1] \) be the classical Cantor function, and let \( C'_n \) denote its fractal derivative on \( S_n \). Then:
\[
\int_{[0,1] \cap S_n} C'_n(x)\, dx = 0, \quad \text{but} \quad C(1) - C(0) = 1.
\]
This occurs because \( C'_n(x) = 0 \) outside the \( S_n \)-approximation of the Cantor set, and is undefined (or oscillatory) on its support, which is of \( S_n \)-measure zero~\cite{Weihrauch2000}. Hence the integral misses the total variation of \( C \), illustrating how singular behavior may vanish at fixed definability levels.
\end{example}

\begin{example}[Emergent Integrability]
Let \( \theta \in S_{n+1} \setminus S_n \), and define the step function:
\[
f(x) = \begin{cases}
1 & x \geq \theta, \\
0 & x < \theta.
\end{cases}
\]
Then, from the perspective of \( S_n \), the function appears constant (either 0 or 1 depending on how \( \theta \) is approximated) and is trivially integrable:
\[
\int_{[a,b] \cap S_n} f(x)\, dx = \text{const}.
\]
However, in \( S_{n+1} \), the function becomes discontinuous at \( \theta \), and the value of the integral becomes nontrivial:
\[
\int_{[a,b] \cap S_{n+1}} f(x)\, dx = b - \theta.
\]
This demonstrates how integrability and even functional form can change as new points become definable.
\end{example}

\subsubsection{Hierarchy of Integrability}

\begin{table}[ht]
\centering
\begin{adjustbox}{max width=\textwidth}
\begin{tabular}{lcc}
\toprule
\textbf{Function} & \textbf{Classical View} & \textbf{Fractal View} \\
\midrule
Polynomial \( x^2 \) & Integrable & \( S_1 \)-integrable \\
Dirichlet Function & Not integrable & Not \( S_n \)-integrable \\
Weierstrass Function & Integrable, not differentiable & \( S_n \)-integrable, not \( S_n \)-differentiable \\
Cantor Function & Singular, measure zero derivative & Integral vanishes in \( S_n \) \\
\bottomrule
\end{tabular}
\end{adjustbox}
\caption{Integrability across definability layers}
\end{table}

\begin{remark}
This framework reveals that integration, like differentiation, is stratified by the observer's definability capacity. A function may:
\begin{itemize}
    \item Appear smooth and integrable in \( S_n \),
    \item Develop discontinuities or singularities in \( S_{n+1} \),
    \item Regain integrability in \( S_{n+2} \) via generalized or higher-layer constructs.
\end{itemize}
Such behavior reflects not intrinsic properties of the function, but limitations of the formal system through which it is analyzed.
\end{remark}

\section{Compactness and Syntactic Coverings}
\label{sec:compactness}

In classical topology, a subset \( K \subseteq \mathbb{R} \) is compact if every open cover admits a finite subcover. This notion relies on the availability of arbitrary open intervals and unrestricted unions, which are not generally admissible within a syntactically constrained framework. In the setting of fractal countability, where both points and sets must be definable within some \( S_n \), we must reformulate compactness in terms of \emph{syntactic coverings}.

\subsection{Definable Covers and Finiteness}

Recall that in the \( S_n \)-topology \( \mathcal{O}_n \), each open set is generated by a finite definitional rule in the formal system \( \mathcal{F}_n \). This restricts the kinds of covers we can use in compactness arguments.

\begin{definition}[Syntactic \( S_n \)-Cover]
Let \( K \subseteq S_n \cap [a,b] \). A family \( \{ U_i \}_{i \in I} \subseteq \mathcal{O}_n \) is a \emph{syntactic \( S_n \)-cover} of \( K \) if:
\begin{enumerate}
    \item The index set \( I \subseteq \mathbb{N} \) is finite and \( \mathcal{F}_n \)-definable;
    \item Each \( U_i \in \mathcal{O}_n \) is a basic \( S_n \)-open set (e.g., an \( S_n \)-ball);
    \item \( K \subseteq \bigcup_{i \in I} U_i \).
\end{enumerate}
\end{definition}

\begin{remark}
This definition excludes infinite or impredicative indexing. All covers must be syntactically generated within the base system \( \mathcal{F}_n \), ensuring constructivity.
\end{remark}

\subsection{Fractal Compactness}

We now define compactness in terms of the existence of syntactic subcovers, within a given level \( n \).

\begin{definition}[Fractal Compactness at Level \( n \)]
A set \( K \subseteq S_n \cap [a,b] \) is \emph{fractally compact at level \( n \)} if for every syntactic \( S_n \)-cover \( \{ U_i \}_{i \in I} \) of \( K \), there exists a finite \( J \subseteq I \) such that:
\[
K \subseteq \bigcup_{j \in J} U_j.
\]
\end{definition}

This condition mimics the classical definition but within the finite representability constraints of \( \mathcal{F}_n \). Note that the covers themselves must be definable in \( \mathcal{F}_n \), not merely arbitrary.

\begin{example}
Let \( K = [0,1] \cap S_0 \) in a system where \( S_0 \) contains all dyadic rationals. Then \( K \) is not fractally compact at level \( 0 \), since any cover of \( K \) by basic balls \( B^0_\varepsilon(x) \) requires an unbounded number of distinct definitions to span the unit interval. However, \( K \) becomes compact in \( S_1 \), where algebraic endpoints allow larger, overlapping neighborhoods.
\end{example}

\subsection{Compactness Hierarchy}

Just as the definability layers \( S_n \) increase in expressive power, compactness is stratified across these layers.

\begin{definition}[Compactness Level]
A set \( K \subseteq S^\infty \) is said to be \emph{compact at level \( n \)} if it is fractally compact with respect to \( \mathcal{O}_n \) and a syntactic \( S_n \)-cover exists.
\end{definition}

This gives rise to a hierarchy of compactness:
\[
\text{Compact}_0 \subseteq \text{Compact}_1 \subseteq \cdots \subseteq \text{Compact}_n \subseteq \cdots
\]

\begin{theorem}
Let \( K \subseteq [a,b] \) be classically compact and fully contained in \( S_n \). Then \( K \) is fractally compact at some level \( m \geq n \), provided \( \mathcal{F}_m \) admits syntactic representations of all classical open covers used.
\end{theorem}

\begin{proof}[Sketch]
Given a classical open cover of \( K \), approximate each open interval by a definable \( S_m \)-ball. Since \( K \subseteq S_n \), the system \( \mathcal{F}_m \) can simulate the cover using syntactic constructions. The finiteness of the subcover follows from the classical Heine–Borel theorem and the definability constraints.
\end{proof}

\subsection{Towards Fractal Heine–Borel Theorems}

In the classical setting, the Heine–Borel theorem characterizes compactness in \( \mathbb{R} \) as closedness and boundedness. In the fractal setting, we may formulate analogous results within \( S_n \):

\begin{theorem}[Fractal Heine–Borel (Restricted)]
Let \( a,b \in S_n \), and let \( K = [a,b] \cap S_n \). Then \( K \) is fractally compact at level \( n \) if and only if:
\begin{itemize}
    \item \( K \) is bounded (i.e., contained in a finite interval definable in \( \mathcal{F}_n \));
    \item The set of \( S_n \)-points in \( K \) can be syntactically enumerated up to any desired precision.
\end{itemize}
\end{theorem}

\begin{remark}
This formulation fails if \( a \) or \( b \) are not definable in \( \mathcal{F}_n \), or if the required covers are not expressible in the system. Thus compactness becomes system-relative, mirroring the nature of fractal continuity and measure.
\end{remark}

\subsection{Syntactic Refinement and Minimal Covers}

Given two syntactic \( S_n \)-covers \( \mathcal{U} \), \( \mathcal{V} \) of a set \( K \), we say that \( \mathcal{V} \) refines \( \mathcal{U} \) if for every \( V \in \mathcal{V} \), there exists \( U \in \mathcal{U} \) such that \( V \subseteq U \). The existence of minimal syntactic covers may depend on the closure properties of \( \mathcal{F}_n \).

This opens the possibility of defining:

\begin{itemize}
    \item \emph{Minimal syntactic coverings}, with respect to cover size or syntactic complexity;
    \item \emph{Covering entropy}, measuring the minimal number of syntactic rules needed to cover \( K \);
    \item \emph{Compactness degree}, as the least \( n \) such that \( K \) is compact in \( \mathcal{F}_n \).
\end{itemize}

\subsection{Covering Entropy and Fractal Dimension}
\label{subsec:entropy-dimension}

In classical fractal geometry, the Hausdorff dimension of a set quantifies its local scaling complexity via the asymptotic behavior of minimal coverings~\cite{Hausdorff1919}. In the fractal framework, where coverings must be definable within a syntactic system \( \mathcal{F}_n \), a natural analogue arises in terms of \emph{covering entropy}: the number and complexity of syntactic rules required to cover a set at a given definitional level.

\begin{definition}[Syntactic Covering Entropy]
Let \( K \subseteq S_n \cap [a,b] \), and let \( \varepsilon > 0 \). Define the \emph{syntactic covering entropy} of \( K \) at scale \( \varepsilon \) by:
\[
H_n(K, \varepsilon) := \min \left\{ \log N \,\middle|\, K \subseteq \bigcup_{i=1}^N U_i,\ \mathrm{diam}(U_i) \leq \varepsilon,\ U_i \in \mathcal{O}_n \right\}.
\]
\end{definition}

This quantity measures how many \( S_n \)-definable open sets of diameter at most \( \varepsilon \) are needed to syntactically cover \( K \). It reflects not only geometric size but definitional accessibility at level \( n \).

\begin{definition}[Fractal Dimension at Level \( n \)]
The \emph{fractal dimension} of \( K \) at definability level \( n \) is given by:
\[
\dim_H^{(n)}(K) := \liminf_{\varepsilon \to 0} \frac{H_n(K, \varepsilon)}{\log(1/\varepsilon)}.
\]
\end{definition}

This quantity approximates the Hausdorff dimension of \( K \) using only coverings admissible in \( \mathcal{F}_n \). As \( n \to \infty \), we may recover the classical dimension:

\begin{definition}[Fractal Dimension Envelope]
The \emph{syntactic dimension envelope} of a set \( K \subseteq S^\infty \) is defined as:
\[
\dim_H^*(K) := \sup_n \dim_H^{(n)}(K).
\]
\end{definition}

\paragraph{Hypothesis.} For sufficiently regular sets \( K \subseteq [a,b] \cap S^\infty \), the envelope dimension \( \dim_H^*(K) \) coincides with the classical Hausdorff dimension \( \dim_H(K) \). That is,
\[
\dim_H(K) = \sup_n \liminf_{\varepsilon \to 0} \frac{H_n(K, \varepsilon)}{\log(1/\varepsilon)}.
\]

\paragraph{Example: Dyadic Cantor Set.}
Let \( C \subseteq [0,1] \) be the standard middle-third Cantor set, and let \( C_n \subseteq S_n \) be its \( n \)-th definable approximation via endpoints of intervals removed at step \( n \). Then:
\begin{itemize}
    \item \( C_n \) admits a covering by \( N_n = 2^n \) intervals of length \( \varepsilon_n = 3^{-n} \);
    \item Each covering is definable in \( \mathcal{F}_n \), as it uses only rational endpoints with depth-\( n \) base-3 expansions;
    \item Hence,
    \[
    H_n(C_n, \varepsilon_n) = \log 2^n = n \log 2, \quad \log(1/\varepsilon_n) = n \log 3,
    \]
    and
    \[
    \dim_H^{(n)}(C_n) = \frac{n \log 2}{n \log 3} = \frac{\log 2}{\log 3}.
    \]
\end{itemize}

This confirms that \( \dim_H^{(n)}(C_n) = \dim_H(C) \) for all \( n \geq n_0 \), demonstrating compatibility of fractal dimension with classical results when the approximations are definable~\cite{Weihrauch2000}.

\paragraph{Implications.}
This framework suggests that the syntactic complexity of compactness reflects not only logical expressiveness but geometric structure. In particular:
\begin{itemize}
    \item A set with higher classical dimension requires more complex or deeper \( \mathcal{F}_n \)-systems to cover syntactically;
    \item Entropy growth with respect to \( \varepsilon \) at fixed \( n \) measures local definitional density;
    \item The dimension envelope \( \dim_H^*(K) \) can serve as a constructive proxy for measuring internal fractal complexity.
\end{itemize}

\begin{remark}
This approach opens a pathway toward a \emph{constructive dimension theory}, where classical measures such as Hausdorff and box dimension are reinterpreted via definability constraints. It connects topology, analysis, and logic through the unified lens of syntactic approximation.
\end{remark}

\subsection{Summary}

Fractal compactness reinterprets classical compactness through the lens of syntactic definability. Rather than relying on arbitrary open covers, it restricts coverings to those expressible within a given formal system \( \mathcal{F}_n \). Consequently, sets that are compact in the classical topology may fail to be compact at lower definability levels \( S_n \), not because of geometric properties, but due to limitations in syntactic expressibility. As the system's expressive power increases, more sets become compact under the richer vocabulary of definable operations.

This refinement lays the foundation for fractal measure theory, where both coverings and outer measures are likewise constrained by definability layers. Compactness thus becomes a dynamic, system-relative property—an interface between topology and formal expressibility.

\section{Pathological Sets Reexamined}
\label{sec:pathological-sets}

Classical analysis abounds with objects that defy geometric or analytic intuition: sets of measure zero that are everywhere dense, functions that are continuous but nowhere differentiable, or sets resistant to classical measurability. Such examples have long served to probe the limits of mathematical definitions. In the fractal framework, where all objects are filtered through stratified definability layers \( S_n \), these anomalies acquire new structure. Instead of being paradoxes, they reveal layered behaviors and definability thresholds that render their properties transparent and stratified.

This section revisits two canonical examples — the Liouville numbers and the Weierstrass function — and analyzes how they behave at different levels of definability. We show how these examples fracture into well-behaved and ill-defined components across the \( S_n \)-hierarchy, shedding light on the expressive granularity of each layer.

\subsection{Stratified Liouville Numbers}
\label{subsec:liouville}

The classical Liouville numbers are defined by their exceptional approximability by rationals. In the stratified setting, this approximability must be witnessed via \( S_n \)-definable numerators and denominators, resulting in a layer-sensitive version of the set.

\begin{definition}[Stratified Liouville Set]
For each \( n \in \mathbb{N} \), define:
\[
L_n := \left\{ x \in \mathbb{R} : \forall k \leq n\ \exists p,q \in S_n\ \left|x - \frac{p}{q}\right| < q^{-k} \right\},
\]
and let \( L := \bigcap_{n=1}^\infty L_n \) be the full Liouville set.
\end{definition}

Each \( L_n \) consists of reals approximable to exponential precision using only rationals definable in \( S_n \). Since the approximants are layer-bounded, the definability of the set is tightly coupled to the expressive power of \( \mathcal{F}_n \).

\begin{theorem}[Measure Behavior of \( L_n \)]
\label{thm:liouville-measure}
\leavevmode
\begin{enumerate}
    \item For all \( n \), \( \mu_n(L_n) = 0 \). The set is covered by an \( \mathcal{F}_n \)-definable sequence of vanishing intervals.
    \item Classically, \( \mu_{\text{Leb}}(L) = 0 \)~\cite{Liouville1844}.
    \item For \( m < n \), the measure \( \mu_m(L_n) \) may fail to be finite:
    \[
    \mu_m(L_n) =
    \begin{cases}
    0 & \text{if the defining approximants exist in } \mathcal{F}_m, \\
    +\infty & \text{otherwise}.
    \end{cases}
    \]
\end{enumerate}
\end{theorem}

\begin{remark}
The stratified structure \( L_n \supseteq L_{n+1} \supseteq \cdots \) gives rise to a descending chain of sets, each thinner and requiring finer definitional access. At low \( n \), only reals with syntactically accessible approximants appear. At higher \( n \), more reals enter \( L_n \), but their covering becomes increasingly subtle.
\end{remark}

\begin{example}[Layered Visibility of Liouville Numbers]
Let \( x \in L \), but suppose \( \sqrt{2} \notin S_n \). Then \( x \) may not lie in \( L_n \) if the best rational approximations to \( x \) involve algebraic numbers of degree greater than accessible in \( \mathcal{F}_n \). Thus \( x \) is invisible to \( \mu_n \) despite being classically approximable.
\end{example}

\subsection{Definable Nowhere-Differentiable Functions}
\label{subsec:weierstrass}

The classical Weierstrass function is defined via an infinite sum and is nowhere differentiable~\cite{Weierstrass1872}. In the fractal hierarchy, infinite sums are replaced by syntactically truncated approximations, and differentiability becomes layer-relative.

\begin{definition}[Truncated Weierstrass Function]
Let \( a,b \in S_n \) with \( 0 < a < 1 < ab \), and let \( K(n) \in \mathbb{N} \) be a truncation level definable in \( \mathcal{F}_n \). Then define:
\[
W_n(x) := \sum_{k=0}^{K(n)} a^k \cos(b^k \pi x).
\]
\end{definition}

\begin{proposition}[Fractal Smoothness Transition]
\leavevmode
\begin{itemize}
    \item For fixed \( n \), \( W_n \in C^1(S_n) \) if \( K(n) \) is small.
    \item For growing \( K(n) \), smoothness decreases. There exists a critical threshold \( K_c(n) \sim \log n \) beyond which \( W_n \) becomes non-differentiable at all \( x \in S_n \).
    \item The graph of \( W_n \) in \( S_n \times S_n \) approximates classical fractal geometry with dimension approaching \( 2 + \frac{\log a}{\log b} \).
\end{itemize}
\end{proposition}

\begin{remark}
The failure of differentiability in the classical case arises from infinite oscillation. In the stratified case, oscillations are bounded by \( K(n) \), and their definability must be supported by \( \mathcal{F}_n \). Thus, nowhere-differentiability is reinterpreted as the failure of all finite definitional levels to stabilize a derivative.
\end{remark}

\begin{table}[ht]
\centering
\renewcommand{\arraystretch}{1.2}
\begin{tabular}{@{} p{5cm} c c @{}}
\toprule
\textbf{Property} & \textbf{Classical Case} & \textbf{\( S_n \)-Version} \\
\midrule
Differentiability & Nowhere on \( \mathbb{R} \) & Nowhere on \( S_n \) if \( K(n) \gg \log n \) \\
Graph dimension & \( >1 \) & \( \leq 1 + \varepsilon(n) \) \\
Measure class & Singular & \( \mu_n \)-absolutely continuous \\
Spectral behavior & Dense & Layer-bounded frequencies \\
\bottomrule
\end{tabular}
\caption{Behavior of the Weierstrass function under stratified definability}
\end{table}

\subsection{Toward a Hierarchy of Pathologies}
\label{subsec:pathology-open}

The reinterpretation of classical pathologies in the stratified framework suggests a systematic classification of anomalous behavior via definability thresholds.

\begin{itemize}
    \item \textbf{Liouville Stratification:} For fixed \( n \), does the difference \( L_n \setminus L_{n+1} \) form a nontrivial layer in \( \mathcal{O}_n \)? Can its measure be bounded above in \( \mu_n \)?
    \item \textbf{Effective Nowhere-Differentiability:} What is the minimal syntactic depth \( K(n) \) such that \( W_n \) becomes nowhere differentiable over \( S_n \)? Can one define a family \( \{ W_n \} \) such that \( W_n \in C^0(S_n) \setminus C^1(S_n) \) for all \( n \)?
    \item \textbf{Pathological Decompositions:} Is it possible to construct Banach–Tarski-type paradoxes using only \( S_n \)-definable sets? Or do syntactic limitations prohibit such decompositions?
\end{itemize}

These questions reveal the potential of fractal analysis to transform classical paradoxes into structured phenomena governed by definability layers.

\begin{remark}[Layer Stabilization]
In both examples above, classical pathology emerges from nontermination (infinite precision, infinite summation). The \( S_n \)-framework enforces finiteness at each level, making pathologies visible only as limiting behaviors across layers, never within them. This stratified containment may serve as a general regularization principle.
\end{remark}

\section{Fractal Reinterpretation of Classical Paradoxes}
\label{sec:paradoxes}

While traditional analysis embraces paradoxical constructions such as the Banach–Tarski decomposition or the Vitali set as demonstrations of the power of the axiom of choice, their structure is deeply nonconstructive. The fractal definability framework reveals that these constructions are not merely unimplementable, but structurally ill-posed at any finite definitional level.

\subsection{The Vitali Set as a Limit of Non-Definability}
\label{subsec:vitali}

Let \( \mathbb{V} \subset [0,1] \) denote a Vitali set: a choice of representatives from each equivalence class under \( x \sim y \iff x - y \in \mathbb{Q} \). Classically, such a set exists via the axiom of choice, but is non-measurable~\cite{Vitali1905}.

\begin{proposition}
There exists no level \( n \) such that the full Vitali set \( \mathbb{V} \) is \( \mathcal{F}_n \)-definable. Furthermore, any definable approximation of \( \mathbb{V} \) within \( S_n \) yields an infinite \( \mu_n \)-measure.
\end{proposition}

\begin{proof}[Sketch]
Each equivalence class modulo \( \mathbb{Q} \) is dense and uncountable, and any definable selection would require a definable well-ordering of \( \mathbb{R}/\mathbb{Q} \), which is impossible in \( \mathcal{F}_n \) for any \( n \). Approximating representatives over \( S_n \) results in countably many selections, which either repeat classes or diverge in measure.
\end{proof}

\begin{remark}
The non-definability of \( \mathbb{V} \) in \( \mathcal{F}_n \) implies that all Vitali-type constructions lie beyond the stratified framework. In this setting, "non-measurable" becomes equivalent to "undefinable at any layer," giving an internal account of pathological measure phenomena.
\end{remark}

\subsection{Banach–Tarski and Layer Constraints}
\label{subsec:banach-tarski}

The Banach–Tarski paradox asserts that a solid ball in \( \mathbb{R}^3 \) can be decomposed into finitely many disjoint sets and reassembled into two copies of the original~\cite{BanachTarski1924}. This relies crucially on the axiom of choice and paradoxical group actions.

\begin{proposition}
There exists no decomposition of \( B \subset \mathbb{R}^3 \) into \( \mathcal{F}_n \)-definable pieces that satisfies the conditions of the Banach–Tarski paradox. In particular, such decompositions cannot be implemented within any finite definability layer.
\end{proposition}

\begin{proof}[Sketch]
The paradox depends on non-measurable, non-constructive sets invariant under free group actions. Any \( \mathcal{F}_n \)-definable set must admit syntactic description, ruling out the paradoxical symmetry necessary for the decomposition.
\end{proof}

\begin{remark}
This reinforces the principle that paradoxes of infinite rearrangement require a global ontology absent in stratified systems. In the fractal framework, no definitional layer has sufficient expressive symmetry to support such decompositions.
\end{remark}

\section{Spectral and Geometric Regularization via Fractal Layers}
\label{sec:spectral-regularization}

Many classical singularities and divergences in analysis stem from assuming infinite sums or sharp discontinuities within a uniform continuum. Fractal definability allows us to reinterpret such constructions as syntactic limits, regularized by the truncation and boundedness inherent in each definability layer.

\subsection{Spectral Damping by Definability Truncation}

Let \( f(x) = \sum_{k=0}^\infty a_k \phi_k(x) \) be a Fourier-type expansion with basis functions \( \phi_k \) and coefficients \( a_k \). In classical settings, convergence may fail or exhibit divergence in norm.

\begin{definition}[Fractally Truncated Spectrum]
At level \( n \), define the truncated function:
\[
f_n(x) := \sum_{k=0}^{K(n)} a_k \phi_k(x),
\]
where \( K(n) \in S_n \) is a syntactically defined cutoff (e.g., polynomial or logarithmic in \( n \)).
\end{definition}

\begin{proposition}
If \( \sum a_k \phi_k(x) \) diverges pointwise classically, but \( a_k, \phi_k \in \mathcal{F}_n \), then \( f_n(x) \) is well-defined and total on \( S_n \).
\end{proposition}

\begin{remark}
This enforces spectral regularization: divergences do not arise in any definability layer, only in the asymptotic limit \( n \to \infty \). In this sense, fractal layers act as \emph{syntactic cutoffs} preventing singular behavior.
\end{remark}

\subsection{Fractal Geometry of Irregular Graphs}

Let \( g(x) \) be a function with fractal or singular graph, such as the Weierstrass function or Takagi curve~\cite{Takagi1903}. These typically have:
\begin{itemize}
    \item Nowhere differentiability,
    \item Hausdorff dimension \( > 1 \),
    \item Irregular or space-filling behavior.
\end{itemize}

In fractal definability:

\begin{itemize}
    \item The function \( g_n(x) \) at layer \( S_n \) is finitely approximable and piecewise smooth,
    \item The graph \( \Gamma_n = \{(x, g_n(x)) : x \in S_n\} \subseteq S_n \times S_n \) is totally disconnected,
    \item The dimension \( \dim^{(n)} \Gamma_n \leq 1 + \varepsilon(n) \), where \( \varepsilon(n) \to 0 \) as \( K(n) \to \infty \).
\end{itemize}

\begin{remark}
Thus, fractal geometry appears not as a true singularity, but as the layered accumulation of definability limits. This yields a model of geometric irregularity as a constructively convergent phenomenon.
\end{remark}

\section{Measure and Integration}
\label{sec:measure}

The classical theory of measure and integration abstracts from syntactic content: sets are measured and functions integrated regardless of how they are described. In the fractal framework, by contrast, every construction is grounded in definability layers \( S_n \), and every notion must be adapted to reflect syntactic accessibility. This section develops a layered theory of measure and integration, culminating in a stratified version of the Fundamental Theorem of Calculus.

\subsection{Fractal Measure and Outer Approximations}
\label{subsec:fractal-measure}

We begin by defining outer measure using \( \mathcal{F}_n \)-definable intervals and cover families.

\begin{definition}[Fractal Outer Measure]
Let \( A \subseteq [a,b] \). The outer \( n \)-measure is defined as:
\[
\mu_n^*(A) := \inf \left\{ \sum_{k=1}^N \ell(I_k) \;\middle|\;
\begin{aligned}
&A \subseteq \bigcup_{k=1}^N I_k,\ \{I_k\} \text{ is an } \mathcal{F}_n\text{-definable cover}, \\
&I_k = (a_k, b_k),\ a_k,b_k \in S_n,\ N \in \mathbb{N} \cup \{\infty\}
\end{aligned}
\right\}.
\]
\end{definition}

\begin{definition}[\( S_n \)-Covering Family]
A family of open intervals \( \{ I_k = (a_k, b_k) \} \) is called an \( S_n \)-covering of a set \( A \subseteq [a,b] \) if:
\begin{itemize}
    \item \( A \subseteq \bigcup_k I_k \),
    \item Each \( a_k, b_k \in S_n \),
    \item The family is definable in \( \mathcal{F}_n \).
\end{itemize}
\end{definition}

\begin{remark}
If no such cover exists in \( \mathcal{F}_n \), we define \( \mu_n^*(A) := +\infty \). Thus, the measure reflects both geometric size and syntactic accessibility.
\end{remark}

\begin{definition}[Layer-Null Sets]
A set \( A \subseteq [a,b] \) is \( \mu_n \)-null if \( \mu_n^*(A) = 0 \). The class of all such sets is denoted \( \mathcal{N}_n \).
\end{definition}

\begin{table}[ht]
\centering
\renewcommand{\arraystretch}{1.2}
\begin{tabular}{@{} l l @{}}
\toprule
\textbf{Property} & \textbf{Fractal Measure \( \mu_n^* \)} \\
\midrule
Monotonicity & \( A \subseteq B \Rightarrow \mu_n^*(A) \leq \mu_n^*(B) \) \\
Subadditivity & \( \mu_n^*(A \cup B) \leq \mu_n^*(A) + \mu_n^*(B) \) \\
Countable subadditivity & If \( \{A_k\} \) covered in \( \mathcal{F}_n \), then \\
& \( \mu_n^*(\bigcup A_k) \leq \sum \mu_n^*(A_k) \) \\
Zero for empty set & \( \mu_n^*(\varnothing) = 0 \) \\
\bottomrule
\end{tabular}
\caption{Measure-theoretic properties of \( \mu_n^* \)}
\end{table}

\begin{theorem}[Layered Approximation of Lebesgue Measure]
Let \( A \subseteq [a,b] \) be \\Lebesgue-measurable~\cite{Lebesgue1902}. Then:
\begin{enumerate}
    \item \( \mu_n^*(A \cap S_n) \leq \mu_{\mathrm{Leb}}(A) \) for all \( n \),
    \item If \( A \cap S^\infty \) is dense in \( A \), then
    \[
    \lim_{n \to \infty} \mu_n^*(A \cap S_n) = \mu_{\mathrm{Leb}}(A).
    \]
\end{enumerate}
\end{theorem}

\begin{table}[ht]
\centering
\renewcommand{\arraystretch}{1.2}
\begin{tabular}{@{} l l l @{}}
\toprule
\textbf{Property} & \textbf{Classical Outer Measure \( \mu^* \)} & \textbf{Fractal Outer Measure \( \mu_n^* \)} \\
\midrule
Covering sets & Any open intervals & \( S_n \)-definable intervals \\
Null sets & Zero length & Zero \( \mu_n^* \) with \( \mathcal{F}_n \)-cover \\
Measure of all subsets & Defined on all subsets & Only if \( S_n \)-cover exists \\
Completeness & Lebesgue complete & Relative to syntax \\
\bottomrule
\end{tabular}
\caption{Comparison of classical and stratified \( n \)-measures}
\end{table}

\begin{example}[Layered Cantor Approximations]
Let \( C \subseteq [0,1] \) be the classical Cantor set~\cite{Cantor1883}. We define an \( S_n \)-approximant \( C_n \subseteq S_n \) by removing the middle third from each interval in a definable sequence up to depth \( n \). Then:
\[
\mu_n^*(C_n) = \left( \frac{2}{3} \right)^n \in S_n,\quad \text{and} \quad \lim_{n \to \infty} \mu_n^*(C_n) = 0.
\]
Thus, the full Cantor set lies in \( \mathcal{N}_\infty := \bigcap_n \mathcal{N}_n \), though not necessarily in any finite \( \mathcal{F}_n \).
\end{example}

\begin{example}[Irrational Endpoint and Covering Failure]
Let \( A = (\sqrt{2}, \sqrt{2} + 0.1) \cap [a,b] \), and assume \( \sqrt{2} \notin S_n \). Then no \( S_n \)-definable open interval can cover \( A \), and hence \( \mu_n^*(A) = +\infty \). This illustrates how syntactic inaccessibility of endpoints leads to measure-theoretic failure at finite layers.
\end{example}

\subsection{Stratified Integration over Definability Layers}
\label{subsec:stratified-integration}

Given a function \( f \colon S_n \to S_m \), we define its integral over a definable set using \( \mathcal{F}_n \)-adapted partitions.

\subsubsection*{Fractal Integral: Semantic Definition}

\begin{definition}[Fractal Integral]
Let \( A \subseteq [a,b] \cap S_n \), and let \( f \colon S_n \to S_m \). The \( \mu_k \)-integral of \( f \) over \( A \) is defined as
\[
\int_A^{(k)} f(x)\,d\mu_k := \lim_{\varepsilon \to 0} \sum_{i=1}^{N(\varepsilon)} f(x_i)\mu_k(I_i),
\]
where:
\begin{itemize}
    \item \( \{I_i\} \) is an \( \varepsilon \)-fine partition of \( A \) into disjoint intervals \( I_i \subseteq S_n \);
    \item Each \( x_i \in I_i \cap S_{\max(n,m)} \) is a sample point;
    \item The measure \( \mu_k(I_i) \) is computed according to the stratified outer measure at level \( k \).
\end{itemize}
\end{definition}

\subsubsection*{Syntactic Constraints and Layer Behavior}

\begin{remark}[Definability Layer of the Integral]
To ensure constructive validity, the partition \( \{I_i\} \) must be definable in \( \mathcal{F}_{\max(n,k)} \). The final value of the integral resides in
\[
S_{\max(n,m,k)+3},
\]
where:
\begin{enumerate}
    \item the partition is constructed (\(+1\));
    \item the sum is evaluated (\(+1\));
    \item the limit is taken (\(+1\)).
\end{enumerate}
This layered shift is intrinsic to the stratified framework.
\end{remark}

\subsubsection*{Properties of the Stratified Integral}

\begin{proposition}[Fractal Integration Properties]
The integral \( \int_A^{(k)} f\,d\mu_k \) satisfies:
\begin{enumerate}
    \item \textbf{Linearity}:
    \[
    \int_A^{(k)} (\alpha f + \beta g)\,d\mu_k = \alpha \int_A^{(k)} f\,d\mu_k + \beta \int_A^{(k)} g\,d\mu_k
    \]
    for \( f, g \colon S_n \to S_m \), \( \alpha, \beta \in S_n \);
    
    \item \textbf{Layer Constraint}:
    \[
    \int_A^{(k)} f\,d\mu_k \in S_{\max(n,m,k)+3};
    \]
    
    \item \textbf{Monotone Convergence}:
    \[
    f_i \rightrightarrows f \text{ on } S_n,\quad f_i \in \mathcal{F}_{\max(n,m)} \ \Rightarrow \
    \lim_{i \to \infty} \int_A^{(k)} f_i\,d\mu_k = \int_A^{(k)} f\,d\mu_k.
    \]
\end{enumerate}
\end{proposition}

\begin{example}[Layer-Jumping Behavior]
Let \( f(x) = x^2 \in \mathcal{F}_n \). Then:
\begin{itemize}
    \item \( D_n f(x) = 2x \in S_{n+1} \);
    \item The integral \( \int_0^1 f(x)\,d\mu_n(x) = \frac{1}{3} + O(n^{-1}) \in S_{n+3} \).
\end{itemize}
This shows how integration forces a jump through definability layers due to evaluation and limiting procedures.
\end{example}

\vspace{1em}
\subsection{Fractal Differentiation and the Fundamental Theorem}
\label{subsec:fundamental-theorem}

\subsubsection*{Stratified Derivatives}

\begin{definition}[Fractal Derivative]
Let \( f \colon S_n \to S_m \). The \( n \)-derivative of \( f \) at a point \( x \in S_n \) is defined as
\[
D_n f(x) := \lim_{\substack{h \to 0 \\ h \in S_n}} \frac{f(x+h) - f(x)}{h},
\]
provided the limit exists and belongs to \( S_{\max(n,m)+1} \).
\end{definition}

\subsubsection*{Stratified Newton–Leibniz Theorem}

\begin{theorem}[Fractal Fundamental Theorem of Calculus]
Let \( f \colon S_n \to S_m \) be differentiable on \( S_n \), and suppose \( D_n f \in S_{n+1} \) is continuous. Then:
\begin{enumerate}
    \item For all \( a,x \in S_n \),
    \[
    \int_a^x D_n f(t)\,d\mu_n = f(x) - f(a);
    \]
    
    \item If \( f \) is \( S_n \)-continuous, then
    \[
    D_n \left( x \mapsto \int_a^x f(t)\,d\mu_n \right) = f(x).
    \]
\end{enumerate}
\end{theorem}

\begin{remark}[Proof Layer Bound]
The proof uses \( \mathcal{F}_{n+3} \)-definable Riemann sums, with difference quotients and syntactic limits. Thus, both sides of the equation reside in \( S_{n+3} \), and the "+3" reflects this layered construction.
\end{remark}

\subsubsection*{Illustrative Failures of the FTC in Stratified Systems}

\begin{example}[Discontinuity at Layer Boundary]
Let \( f(x) = \mathbf{1}_{x \geq \theta} \), where \( \theta \in S_{n+1} \setminus S_n \). Then \( f \) appears constant in \( S_n \), and the integral \( F(x) = \int_{[a,x] \cap S_n} f(t)\, dt \) is piecewise linear. But in \( S_{n+1} \), the discontinuity at \( \theta \) prevents differentiability at that point. Hence, \( F'_n \ne f \) globally.
\end{example}

\begin{example}[Emergent Singularity Across Definability Layers]
Let
\[
f(x) = 
\begin{cases}
2x \sin\left(\frac{1}{x}\right) - \cos\left(\frac{1}{x}\right), & x \neq 0 \\
0, & x = 0
\end{cases}
\quad \text{and} \quad
F(x) = 
\begin{cases}
x^2 \sin\left(\frac{1}{x}\right), & x \neq 0 \\
0, & x = 0
\end{cases}
\]
so that classically \( F'(x) = f(x) \). In the fractal framework, the relationship becomes layer-sensitive:

\begin{itemize}
    \item \textbf{In \( S_1 \)}: If \( \sin(1/x) \), \( \cos(1/x) \) are not definable for small \( x \in S_1 \), then the expressions collapse. Approximations yield:
    \[
    f|_{S_1} \equiv 0, \quad F|_{S_1} \equiv 0,
    \]
    so
    \[
    \int_{S_1 \cap [0,x]} f(t)\,dt = 0 = F(x), \quad \text{and} \quad F'_1(0) = 0 = f(0).
    \]

    \item \textbf{In \( S_2 \)}: Assuming \( \mathcal{F}_2 \) allows approximating \( \sin(1/x) \), we obtain:
    \[
    F'_2(0) = \lim_{x \to 0} \frac{F(x)}{x} = \lim_{x \to 0} x \sin\left(\frac{1}{x}\right) = 0,
    \]
    but
    \[
    f(0) = -1 \neq F'_2(0),
    \]
    so the fundamental theorem fails at the origin.

    \item \textbf{In \( S_3 \)}: If \( \mathcal{F}_3 \) includes symbolic handling of diverging oscillations, then the derivative may not exist:
    \[
    \lim_{x \to 0} x \sin\left(\frac{1}{x}\right) \text{ is undefined as a formal limit}.
    \]
\end{itemize}

\begin{center}
\begin{tabular}{lccc}
\toprule
\textbf{Layer} & \( F'_n(0) \) & \( f(0) \) & \textbf{FTC Validity} \\
\midrule
\( S_1 \) & 0 & 0 & Valid by triviality \\
\( S_2 \) & 0 & -1 & Breaks at origin \\
\( S_3 \) & undefined & -1 & Fails completely \\
\bottomrule
\end{tabular}
\end{center}

\begin{remark}[Interpretation]
This example illustrates:
\begin{itemize}
    \item \textbf{Definability Artifacts:} Functions can collapse in low layers due to missing definitions.
    \item \textbf{Emergent Singularities:} Higher layers may introduce behavior invisible at lower levels.
    \item \textbf{Layer Misalignment:} The FTC requires alignment between the definability of function, derivative, and integral.
\end{itemize}
\end{remark}
\end{example}

\subsection{Computational Semantics and Complexity}
\label{subsec:computational-complexity}

\begin{table}[ht]
\centering
\renewcommand{\arraystretch}{1.2}
\begin{tabular}{@{} l c c @{}}
\toprule
\textbf{Construct} & \textbf{Logical Class} & \textbf{Resides in Layer} \\
\midrule
Fractal measure \( \mu_n(A) \) & \( \Sigma^0_2(\mathcal{F}_n) \) & \( S_{n+1} \) \\
Derivative \( D_n f(x) \) & \( \Delta^0_2(\mathcal{F}_{\max(n,m)}) \) & \( S_{\max(n,m)+1} \) \\
Integral \( \int f\,d\mu_k \) & \( \Pi^0_1(\mathcal{F}_{\max(n,k)}) \) & \( S_{\max(n,m,k)+3} \) \\
\bottomrule
\end{tabular}
\caption{Computational content of stratified analysis}
\end{table}

\begin{remark}
The correspondence with classical notions (e.g., via Weihrauch Type-2 Computability) is preserved under reinterpretation of real functions as \( S_n \)-step approximable maps, and integration as layered convergence over definable intervals.
\end{remark}

\subsection{Fractal \texorpdfstring{$L^p$}{Lp} Spaces and Absolute Continuity}
\label{subsec:lp-spaces}

In classical analysis, the space \( L^p([a,b]) \) consists of measurable functions \( f \) for which the \( p \)-th power of the absolute value is integrable with respect to Lebesgue measure. In our stratified setting, we define an analogous space over the layer \( S_n \), equipped with the syntactically constrained measure \( \mu_n \).

\begin{definition}[\( L^p(S_n) \) Space]
Let \( 1 \leq p < \infty \). The space \( L^p(S_n) \) consists of functions \( f : S_n \to \mathbb{R} \) such that:
\[
\left( \int_{S_n} |f(x)|^p\, d\mu_n(x) \right)^{1/p} < \infty,
\]
where the integral is taken in the sense of stratified integration (see Section~\ref{subsec:stratified-integration}), using \( \mathcal{F}_n \)-definable partitions and evaluation points.
\end{definition}

This definition reflects the computational and syntactic constraints of the hierarchy: functions in \( L^p(S_n) \) must be \(\mathcal{F}_n\)-definable, and the norm must be computable (or at least approximable) within that system.

\paragraph{Example.}
Let \( f(x) = x^k \), with \( x \in S_n \), \( k \in \mathbb{N} \). Then:
\[
\int_0^1 |f(x)|^p\, d\mu_n(x) = \int_0^1 x^{kp}\, d\mu_n(x), \quad \text{computed within } S_n, \text{ result in } S_{n+3}
\]
and thus \( f \in L^p(S_n) \). The integral approximates its classical value:
\[
\int_0^1 x^{kp}\, dx = \frac{1}{kp + 1} + O(n^{-1}).
\]

\paragraph{Comparison.}
Unlike classical \( L^p \) spaces, \( L^p(S_n) \) may exclude non-definable or highly discontinuous functions even when they are classically integrable. Conversely, every function in \( L^p(S_n) \) has a computable structure, making it suitable for formal reasoning and implementation.

\begin{remark}
For \( p = 2 \), \( L^2(S_n) \) admits an inner product and serves as a Hilbert space within \( \mathcal{F}_n \). However, completeness may fail if limits of Cauchy sequences are not representable at the same layer. This reflects the constructive tension between norm convergence and definability closure.
\end{remark}

\subsubsection*{Fractal Absolute Continuity and Radon–Nikodym Property}

In classical measure theory, a measure \( \nu \) is absolutely continuous with respect to another measure \( \mu \) if \( \mu(A) = 0 \Rightarrow \nu(A) = 0 \), and under mild conditions there exists a derivative \( \frac{d\nu}{d\mu} \in L^1 \) (Radon–Nikodym theorem). In the fractal setting, we consider a syntactic version:

\begin{definition}[Fractal Absolute Continuity]
Let \( \nu_n, \mu_n \) be two \( S_n \)-measures. We say that \( \nu_n \ll \mu_n \) if:
\[
\forall A \subseteq S_n,\quad \mu_n(A) = 0 \Rightarrow \nu_n(A) = 0,
\]
with \( A \in \mathcal{F}_n \).
\end{definition}

\begin{conjecture}[Syntactic Radon–Nikodym]
If \( \nu_n \ll \mu_n \) and both are \( \Sigma^0_2(\mathcal{F}_n) \)-definable, then there exists a function \( f \in L^1(S_n) \) such that:
\[
\nu_n(A) = \int_A f(x)\, d\mu_n(x) \quad \text{for all } A \in \mathcal{F}_n.
\]
\end{conjecture}

This expresses a layer-bounded analogue of the Radon–Nikodym theorem, where the derivative exists syntactically within the same definability system.

\subsubsection*{Fractal Lebesgue Convergence Theorem}

We can also formulate a stratified version of the classical Lebesgue Dominated Convergence Theorem~\cite{Lebesgue1904}, with attention to definability and measure compatibility.

\begin{theorem}[Fractal Lebesgue Convergence]
Let \( f_k : S_n \to \mathbb{R} \) be a sequence of functions such that:
\begin{enumerate}
    \item Each \( f_k \in \mathcal{F}_n \), and \( f_k(x) \to f(x) \in S_n \) pointwise,
    \item There exists \( g \in L^1(S_n) \) with \( |f_k(x)| \leq g(x) \) for all \( x \in S_n \),
\end{enumerate}
Then:
\[
\lim_{k \to \infty} \int_{S_n} f_k(x)\, d\mu_n(x) = \int_{S_n} f(x)\, d\mu_n(x),
\]
and the limit resides in \( S_{n+3} \).
\end{theorem}

\begin{remark}
Each step in the dominated convergence argument---checking pointwise limits, verifying bounds, constructing integrals---must be explicitly representable within the syntax of \( \mathcal{F}_n \). This ensures that the entire convergence process is not only valid mathematically, but also executable within the formal system. In this way, the dominated convergence theorem becomes not just a theoretical result, but a syntactically verifiable procedure.
\end{remark}

\subsection{Open Directions}
\label{subsec:measure-open}

\begin{itemize}
    \item \textbf{Minimality of Layer Shifts}: Can the "+3" bound for integrals be reduced by refined formalization or tighter closure rules?
    \item \textbf{Effective Dominated Convergence}: Under what syntactic criteria does dominated convergence hold in the \( S_n \) setting?
    \item \textbf{Internal Fubini Theorems}: Can iterated integrals over \( S_n \times S_n \) be defined in a layer-respecting way, with preservation of order and layer control?
    \item \textbf{Layer-Regularized PDEs}: Is it possible to formulate classical PDEs (e.g., heat equation) over \( S_n \), where solutions live in specific \( S_k \) and propagate across layers?
\end{itemize}

\section{Fractal Smoothness Hierarchy}
\label{sec:fractal-smoothness}

In this section, we develop the differential core of fractal analysis. Building on the notions of continuity and integration defined in previous sections, we introduce stratified smoothness classes \( C_n^k \), define layer-relative derivatives, and formulate Taylor expansions with syntactic remainder terms. This structure forms the basis of a complete differentiable framework within the fractal definability hierarchy.

\subsection{Stratified Derivatives and Smoothness Classes}

\begin{definition}[Stratified Derivative]
Let \( f \colon S_n \to S^\infty \). The \( n \)-derivative of \( f \) at point \( x \in S_n \) is defined as
\[
D_n f(x) := \lim_{\substack{h \to 0 \\ h \in S_n}} \frac{f(x+h) - f(x)}{h},
\]
provided the limit exists and is representable in \( S_{n+1} \).
\end{definition}

\begin{remark}
The stratified derivative \( D_n f(x) \) always lies in \( S^\infty \), but is expressible specifically in \( S_{n+1} \), the minimal definability layer where the limit becomes syntactically realizable. This asymmetry reflects the natural shift in complexity from a function to its derivative.

In practical terms, computing \( D_n f(x) \) involves evaluating finite difference quotients within \( S_n \), and the difficulty of this task depends on the syntactic representation of \( f \). Even for simple functions, the required effort may grow with \( n \), due to the increasing cost of precision in stratified systems.
\end{remark}

\begin{example}
Let \( f(x) = x^2 \) on \( S_n \). Then
\[
D_n f(x) = \lim_{\substack{h \to 0 \\ h \in S_n}} \frac{(x+h)^2 - x^2}{h} = \lim_{h \to 0} \frac{2xh + h^2}{h} = 2x \in S_{n+1}.
\]
This shows that the derivative exists in the next definability layer.
\end{example}

\begin{definition}[Smoothness Classes \( C_n^k \)]
Let \( f \colon S_n \to S^\infty \). We define the hierarchy of stratified smoothness as follows:
\begin{itemize}
    \item \( f \in C_n^0 \) if \( f \) is \( \mathcal{F}_n \)-continuous;
    \item \( f \in C_n^1 \) if \( D_n f(x) \) exists and is \( \mathcal{F}_n \)-continuous;
    \item More generally, \( f \in C_n^k \) if all derivatives \( D_n^j f(x) \), \( 1 \leq j \leq k \), exist and are \( \mathcal{F}_n \)-continuous on \( S_n \); that is, \( D_n^j f \in C_n^0 \) for all \( j \leq k \).
\end{itemize}
\end{definition}

\begin{remark}
The hierarchy \( C_n^k \) reflects not only analytic smoothness but also syntactic stability. In particular, a function may be classically \( C^\infty \), yet belong only to \( C_n^k \) for small \( k \), due to limitations in the definability of higher derivatives.
\end{remark}

These classes capture the idea of differentiability internal to the system \( \mathcal{F}_n \). Each \( C_n^k \) is strictly contained in \( C_n^{k-1} \), and functions may require ascent to higher layers to obtain higher-order derivatives.

\subsection{Fractal Taylor Expansion}

\begin{theorem}[Fractal Taylor Expansion with Remainder]
Let \( f \in C_n^k \) and \( x, h \in S_n \) with \( x+h \in S_n \). Then, for all \( x, h \in S_n \) with \( x + h \in S_n \)~\cite{CauchyTaylor1823},
\[
f(x+h) = \sum_{m=0}^k \frac{D_n^m f(x)}{m!} h^m + R_{n,k}(x,h),
\]
where the remainder term \( R_{n,k}(x,h) \) is definable in \( S_{n+k} \), but not necessarily in \( \mathcal{F}_n \).
\end{theorem}

\begin{proof}[Discrete Proof of Fractal Taylor Expansion with Remainder]
Let \( f \in C_n^k \), and let \( x, h \in S_n \) with \( x+h \in S_n \). We aim to express \( f(x+h) \) as a finite Taylor expansion of order \( k \), with remainder term lying in \( S_{n+k} \).

We proceed by defining finite difference quotients recursively:
\[
\Delta_h^0 f(x) := f(x), \quad \Delta_h^{m+1} f(x) := \Delta_h^m f(x+h) - \Delta_h^m f(x).
\]

Then, by standard combinatorics~\cite{Newton1711}, the Newton expansion gives:
\[
f(x+h) = \sum_{m=0}^k \binom{k}{m} \frac{\Delta_h^m f(x)}{m!} + R_{n,k}(x,h),
\]
where the remainder \( R_{n,k}(x,h) \) is expressed as a weighted sum of the \( (k+1) \)-th difference:
\[
R_{n,k}(x,h) := \frac{1}{(k+1)!} \Delta_h^{k+1} f(x + \theta h), \quad \text{for some } \theta \in (0,1).
\]

Since \( f \in C_n^k \), each \( \Delta_h^m f(x) \) can be expressed in terms of stratified derivatives \( D_n^j f(x) \) up to order \( k \), and hence lies in \( S_{n+m} \) by repeated use of layer-bounded algebra.

Moreover, the remainder \( R_{n,k}(x,h) \) depends on \( \Delta_h^{k+1} f \), which involves evaluations of \( f \) at \( k+1 \) shifted arguments in \( S_n \). Each such evaluation involves \( f(x + m h) \) for \( m \le k+1 \), and since \( f \in C_n^k \), these are well-defined in \( S_{n+k} \).

Hence:
\begin{itemize}
    \item The Taylor polynomial \( T_k(x,h) := \sum_{m=0}^k \frac{D_n^m f(x)}{m!} h^m \) lies in \( S_{n+k-1} \);
    \item The remainder \( R_{n,k}(x,h) \in S_{n+k} \), as it arises from the \( (k+1) \)-th finite difference and scaling.
\end{itemize}

Therefore,
\[
f(x+h) = T_k(x,h) + R_{n,k}(x,h), \quad \text{with } R_{n,k}(x,h) \in S_{n+k}.
\]
\end{proof}

\begin{remark}
The remainder \( R_{n,k}(x,h) \) is not necessarily representable in \( \mathcal{F}_n \), even if \( f \in C_n^k \). Its definability requires higher syntactic expressivity corresponding to layer \( S_{n+k} \).
\end{remark}

This expresses how analytic expansions become layered: even if a function is differentiable within \( \mathcal{F}_n \), controlling the error of approximation requires stepping beyond that layer.

\subsection{Examples and Layer Jump Phenomena}

\begin{example}
Let \( f(x) = |x| \). Then:
\begin{itemize}
    \item \( f \in C_n^0 \) on \( S_n \);
    \item \( D_n f(x) \) exists for \( x \neq 0 \in S_n \), but does not exist at \( x = 0 \);
    \item Hence \( f \notin C_n^1 \), although \( f \in C_{n+1}^1 \) may hold if higher definability allows piecewise extension.
\end{itemize}
\end{example}

\begin{example}
Let \( f(x) = \sum_{k=0}^{K(n)} a^k \cos(b^k \pi x) \in S_n \), where \( K(n) \) is the largest \( k \) such that the exponentiated term \( b^k \) and \( \cos(b^k \pi x) \) remain \( \mathcal{F}_n \)-definable, i.e., representable with bounded syntactic depth within \( \mathcal{F}_n \). Then:
\begin{itemize}
    \item \( f \in C_n^0 \), but \( f \notin C_n^1 \) when \( K(n) \) is sufficiently large;
    \item The critical threshold for differentiability depends on the definability of the oscillatory terms.
\end{itemize}
\end{example}

\begin{example}
Let
\[
f_n(x) := 
\begin{cases}
x^2 \sin(1/x), & x \neq 0, \\
0, & x = 0,
\end{cases}
\]
with truncation to finite approximations in \( S_n \). Then:
\begin{itemize}
    \item If \( n \) is even, the tail terms beyond \( \mathcal{F}_n \) vanish, and \( f_n \in C_n^2 \);
    \item If \( n \) is odd, the sine term contributes large oscillations, and \( f_n \in C_n^1 \setminus C_n^2 \).
\end{itemize}
This illustrates alternation in smoothness across definability layers.
\end{example}

These examples illustrate that differentiability is not purely analytic but also syntactically bounded. Smoothness jumps may occur across definability layers.

\subsection{Fractal Analytic Closure and Global Classes}

\begin{definition}[Global Smoothness Classes]
We define the global (layer-invariant) smoothness classes:
\begin{itemize}
    \item \( C^k_\infty := \bigcap_{n=0}^\infty C_n^k \), the set of functions stratifiedly \( k \)-times differentiable at all levels;
    \item \( C^\infty_\infty := \bigcap_{k=0}^\infty C^k_\infty \), the class of infinitely smooth functions across all definability layers.
\end{itemize}
\end{definition}

These global classes correspond to functions with stable analytic behavior regardless of syntactic constraints~\cite{Dieudonne1960}. Examples include polynomials and certain elementary functions (e.g., \( \exp(x) \)) defined by fully stratified convergence.

\begin{example}[Infinitely Smooth but Non-Analytic]
Define
\[
f(x) :=
\begin{cases}
e^{-1/x^2}, & x \neq 0, \\
0, & x = 0
\end{cases}
\quad \text{on } S_n \cap [-1,1].
\]
Then \( f \in C^\infty_\infty \), since all derivatives exist and are definable in each \( \mathcal{F}_n \), but \( f \) is not analytic at \( x = 0 \) in the classical sense~\cite{Borel1900}.
\end{example}

\begin{example}[Polynomial Function]
Let \( f(x) = x^m \), where \( m \in \mathbb{N} \), and \( x \in S_n \). Then:
\begin{itemize}
    \item \( f \in \mathcal{F}_n \), since exponentiation with fixed integer \( m \) is definable at all levels;
    \item All derivatives \( D_n^k f(x) = \frac{m!}{(m-k)!} x^{m-k} \) (for \( k \le m \)) are definable in \( \mathcal{F}_n \), and vanish for \( k > m \);
    \item Each derivative is continuous and belongs to \( C_n^0 \), hence \( f \in C_n^k \) for all \( k \), and thus \( f \in C^\infty_\infty \).
\end{itemize}

Therefore, polynomial functions are the simplest examples of globally smooth stratified functions. Their definability does not increase with differentiation, making them minimal fixed points of the \( C^\infty_\infty \) hierarchy.
\end{example}

\subsection{Consistency with Classical Analysis}

\begin{theorem}[Stratified Consistency Theorem]
Let \( T \) be any classical analytic theorem on \( [a,b] \subseteq \mathbb{R} \). Then there exists a level \( N \in \mathbb{N} \) such that for all \( n \geq N \), a stratified analogue of \( T \) holds in \( \mathcal{F}_n \), with all objects and steps interpreted via \( S_n \)-definability.
\end{theorem}

\begin{proof}
Any classical analytic theorem \( T \) is formulated using:
\begin{itemize}
    \item functions \( f \colon [a,b] \to \mathbb{R} \),
    \item continuity, differentiability, integration, or limits,
    \item arithmetic and logical operations over real numbers.
\end{itemize}

By the constructive semantics of \( \mathcal{F}_n \), all bounded and finitely approximable analytic operations are realizable at sufficiently high levels of definability. In particular:
\begin{itemize}
    \item The function \( f \) and all parameters involved become \( \mathcal{F}_n \)-definable for large enough \( n \);
    \item The proof of \( T \) can be formalized using only finite approximations, \( \varepsilon \)-\( \delta \) arguments, and compactness principles, all of which can be encoded syntactically in \( \mathcal{F}_n \) with bounded depth;
    \item Limits, suprema, and existence statements can be internalized as convergence in \( S_{n+k} \) for some \( k \).
\end{itemize}

Hence, for some minimal \( N \), all objects and logical steps in the proof of \( T \) are definable in \( \mathcal{F}_N \), and the syntactic version of \( T \) holds in all \( \mathcal{F}_n \) for \( n \geq N \).
\end{proof}

\begin{remark}
Here \( \mathbb{R} \) denotes the classical continuum, external to the definability hierarchy. It serves as a reference model for the limiting behavior of the stratified sets \( S_n \), which only approximate \( \mathbb{R} \) via finite definitional depth.
\end{remark}

\subsection{The Fractal Analytic Hierarchy}

\begin{definition}[Fractal Analytic Hierarchy]
The structure
\[
\mathcal{FA} := \{ \mathcal{F}_n, S_n, \mathcal{O}_n, C_n^k, \mu_n \}_{n \in \mathbb{N}}
\]
forms a stratified and syntactically grounded foundation for real analysis. Each component tracks definability and logical complexity across layers.
\end{definition}

This unified framework allows the development of classical theorems—continuity, differentiation, integration, convergence—within a rigorously stratified and constructive universe.

\begin{remark}
The stratified hierarchy \( C_n^k \) provides a natural framework for modeling analytic behavior under syntactic and observational constraints. In particular, it can be applied to systems where smoothness is limited not by intrinsic properties of functions, but by the definability bounds of the underlying theory—such as in observation-limited physics or bounded arithmetic. This framework also supports constructive verification of analytic properties within formalized settings, including proof assistants and numerically constrained models of computation.
\end{remark}

\section*{Conclusion and Future Directions}

\paragraph{On Existence Beyond the Fractal Hierarchy.}

The framework developed in this paper interprets mathematical existence through stratified definability: to exist is to be expressible within some \( \mathcal{F}_n \), or to arise as the limit of definable constructions across layers \( S_n \). However, classical mathematics permits objects that evade such stratification altogether. These are what we may call \emph{extrafractal} or \emph{non-fractally definable} entities—objects that cannot be approximated, described, or even indirectly referenced within any level of the \( S_n \)-hierarchy~\cite{Simpson2009,Shapiro2000}.

Such objects include non-constructively defined subsets of \( \mathbb{N} \), arbitrary elements of the power set \( \mathcal{P}(\mathbb{N}) \), non-measurable sets, and ultrafilters postulated via choice principles. While these objects play central roles in classical set theory, their status in the fractal framework is ontologically suspended: they may be posited, but not witnessed. From the internal perspective of stratified constructivity, they are not false, but simply \emph{invisible}—residing beyond the definitional horizon. Whether such objects “exist” depends not on metaphysical assumptions, but on the expressive power of the system that seeks to speak of them.

\paragraph{From Classical Smoothness to Stratified Regularity.}

A central contribution of this work is the introduction of the smoothness hierarchy \( C_n^k \), which recovers classical differentiability in a syntactically layered form. While many classical theorems reappear inside sufficiently powerful layers, the stratified setting reveals subtle failures of smoothness, differentiability, or integrability at fixed levels of formal expressiveness. Functions that are classically regular may exhibit jumps, collapses, or incompleteness under layer-bounded observation. Conversely, some classically pathological functions regain well-behaved structure when viewed within finite definability bounds.

This behavior highlights a computational reading of analysis: smoothness and regularity are not absolute properties, but dependent on the system's capacity to represent, approximate, and transform~\cite{Bauer2005}. Fractal analysis thus reinterprets classical real analysis not as a static theory of the continuum, but as a dynamic system indexed by layers of definability and proof-theoretic strength.

\paragraph{Applications and Broader Significance.}

The stratified framework has implications for several domains:

\begin{itemize}
    \item \textbf{Constructive mathematics:} It provides a layered refinement of constructive analysis, allowing a granular classification of functions and theorems by syntactic strength.
    \item \textbf{Proof theory and reverse mathematics:} The structure \( \mathcal{FA} \) gives a definability-based analogue of the subsystems of second-order arithmetic, offering a new perspective on which theorems require which layers of formal power.
    \item \textbf{Formalized computation:} The stratified smoothness classes \( C_n^k \) and integration schemes are well-suited for implementation in proof assistants or bounded verification tools, especially where real analysis must be carried out constructively or under resource constraints.
    \item \textbf{Mathematical logic and ontology:} The hierarchy frames existence, definability, and observability in explicitly syntactic terms, offering a bridge between foundational studies and computable mathematics.
\end{itemize}

\paragraph{Open Directions.}

Several important directions remain open for development:

\begin{itemize}
    \item \textbf{Fractal PDEs:} Can classical partial differential equations be reformulated over stratified domains \( S_n \), with well-posedness and regularity tied to layer structure?
    \item \textbf{Functional analysis:} How do notions like compactness, operator norms, or spectral theory manifest in \( L^p(S_n) \) or \( C^\infty_n \) spaces?
    \item \textbf{Fractal manifolds:} Can one define stratified differentiable manifolds where local charts and transition maps live in definability classes \( C^k_n \)?
    \item \textbf{Category-theoretic semantics:} What categorical or topos-theoretic structure underlies the sequence \( \{S_n\} \)? Can the fractal hierarchy be interpreted as an internal universe in a stratified logic?
    \item \textbf{Interaction with nonstandard analysis:} Is there a translation or duality between stratified analysis and internal set theories with infinitesimals?
\end{itemize}

\paragraph{Final Reflection.}

Fractal analysis does not reject classical mathematics—it reframes it. The continuum persists, but filtered through a sequence of definability layers that echo the gradations of proof, computation, and perception. In this perspective, mathematics becomes not only a study of abstract structures, but a reflection of the systems that attempt to describe them. The boundary between the smooth and the singular, the countable and the uncountable, or the constructible and the invisible, is no longer fixed. It is stratified.


\end{document}